\documentclass[11pt]{article}

\usepackage[margin=1.125in]{geometry}
\usepackage[dvipsnames]{xcolor}
\definecolor{DarkGreen}{rgb}{0.1,0.5,0.1}
\definecolor{DarkRed}{rgb}{0.5,0.1,0.1}
\definecolor{DarkBlue}{rgb}{0.1,0.1,0.5}

\usepackage{mleftright}

\usepackage[small]{caption}
\usepackage[pdftex]{hyperref}
\hypersetup{
    unicode=false,          
    pdftoolbar=true,        
    pdfmenubar=true,        
    pdffitwindow=false,      
    pdfnewwindow=true,      
    colorlinks=true,       
    linkcolor=DarkBlue,          
    citecolor=DarkGreen,        
    filecolor=DarkGreen,      
    urlcolor=DarkBlue,          
    pdftitle={},
    pdfauthor={},    
}
\usepackage{amssymb,amsmath,amsthm,amsfonts,enumitem, dsfont}
\usepackage{fullpage,comment,mathtools}
\usepackage{tikz}
\usetikzlibrary{arrows,decorations.pathmorphing,decorations.shapes,snakes,patterns,positioning}
\usepackage[ruled,linesnumbered]{algorithm2e}
\usepackage{thm-restate}
\mathtoolsset{showonlyrefs}

\newcommand{\cC}{\ensuremath{\mathcal{C}}} 
  
\newcommand{\cM}{\ensuremath{\mathcal{M}}}

\newcommand{\R}{{\mathbb R}}

\newcommand{\F}{{\mathbb F}}

\newcommand{\N}{\ensuremath{\mathbb{N}}}

\newcommand{\PR}[1]{\mathrm{Pr}\left[ #1\right]}
\renewcommand{\Pr}{\mathrm{Pr}}
\newcommand{\E}[1]{\mathbb{E}\left[ #1\right]}

\newcommand{\Eover}[2]{\mathop{\mathbb{E}}_{#1}\left[ #2 \right]}
\newcommand{\Var}[1]{\mathrm{Var}\left( #1\right)}

\newcommand{\PRover}[2]{\mathop{\mathrm{Pr}}_{#1}\left[ #2\right]}

\newcommand{\ceil}[1]{\lceil {#1}\rceil}
\newcommand{\floor}[1]{\left\lfloor {#1}\right\rfloor}

\newcommand{\inabs}[1]{\left|#1\right|}

\newcommand{\inset}[1]{\left\{#1\right\}}

\newcommand{\inparen}[1]{\left(#1\right)}

\newcommand{\suchthat}{\,:\,}

\newcommand{\supp}{\mathrm{supp}}

\newcommand{\spn}{\ensuremath{\operatorname{span}}}

\newcommand{\ind}[1]{^{(#1)}}

\newcommand{\eps}{\varepsilon}
\renewcommand{\epsilon}{\varepsilon}

\DeclareMathOperator*\Bernoulli{Bernoulli}
\DeclareMathOperator*\Ber{Bernoulli}
\DeclareMathOperator*\Uni{Uniform}

\newcommand{\wt}[1]{\mathrm{wt}\mleft(#1\mright)}

\renewcommand{\exp}[1]{\mathrm{exp}\mleft(#1\mright)}

\theoremstyle{plain}
\declaretheorem[name=Theorem,numberwithin=section]{theorem}
\declaretheorem[name=Lemma,sibling=theorem]{lemma}
\newtheorem*{lemma*}{Lemma} 
\newtheorem*{theorem*}{Theorem} 
\newtheorem{definition}[theorem]{Definition}

\newtheorem{remark}[theorem]{Remark}

\newtheorem{proposition}[theorem]{Proposition}

\newcommand\defeq{\ensuremath{\stackrel{\rm def}{=}}} 

\usepackage{authblk}
\title{Bounds for list-decoding and list-recovery of random linear codes\thanks{
RL, SS and MW are partially funded by NSF-CAREER grant CCF-1844628, NSF-BSF grant CCF-1814629, and a Sloan Research Fellowship. RL is partially supported by NSF GRFP grant DGE-1656518. SS is partially supported by a Google Graduate Fellowship. VG, JM, and NR are partially funded by NSF grants CCF-1563742 and CCF-1814603. NR is partially supported by NSF grants CCF-1527110, CCF-1618280, CCF-1814603, CCF-1910588, NSF CAREER award CCF-1750808 and a Sloan Research Fellowship.}
}

\author[2]{Venkatesan Guruswami}
\author[1]{Ray Li}
\author[2]{Jonathan Mosheiff} 
\author[2]{Nicolas Resch}
\author[1]{Shashwat Silas}
\author[1]{Mary Wootters}
\affil[1]{Stanford University}
\affil[2]{Carnegie Mellon University}

\date{April, 2020}

\begin{document}

\maketitle

\vspace{-5ex}
\begin{abstract}
A family of error-correcting codes is list-decodable from error fraction $p$ if, for every code in the family, the number of codewords in any Hamming ball of fractional radius $p$ is less than some integer $L$ that is independent of the code length. It is said to be list-recoverable for input list size $\ell$ if for every sufficiently large subset of codewords (of size $L$ or more), there is a coordinate where the codewords take more than $\ell$ values. The parameter $L$ is said to be the ``list size" in either case. The capacity, i.e., the largest possible rate for these notions as the list size $L \to \infty$, is known to be $1-h_q(p)$ for list-decoding, and $1-\log_q \ell$ for list-recovery, where $q$ is the alphabet size of the code family. 

In this work, we study the list size of \emph{random linear codes} for both list-decoding and list-recovery as the rate approaches capacity. We show the following claims hold with high probability over the choice of the code (below $q$ is the alphabet size, and $\epsilon > 0$ is the gap to capacity).

\begin{itemize}
\vspace{-1ex}
\itemsep=0ex
    \item A random linear code of rate $1 - \log_q(\ell) - \epsilon$ requires list size $L \ge \ell^{\Omega(1/\epsilon)}$ for list-recovery from input list size $\ell$. This is surprisingly in contrast to completely random codes, where $L = O(\ell/\epsilon)$ suffices w.h.p.

\item  A random linear code of rate $1 - h_q(p) - \epsilon$ requires list size  $L \ge \floor{h_q(p)/\epsilon+0.99}$ for list-decoding from error fraction $p$, when $\epsilon$ is sufficiently small.

\item  A random \emph{binary} linear 
code of rate $1 - h_2(p) - \epsilon$ is list-decodable from \emph{average} error fraction $p$ with list size with $L \leq \floor{h_2(p)/\eps} + 2$. (The average error version measures the average Hamming distance of the codewords from the center of the Hamming ball, instead of the maximum distance as in list-decoding.)
\end{itemize}

\vspace{-1ex}
The second and third results together precisely pin down the list sizes for binary random linear codes for both list-decoding and average-radius list-decoding to three possible values.

Our lower bounds follow by exhibiting an explicit subset of codewords so that this subset---or some symbol-wise permutation of it---lies in a random linear code with high probability.  This uses a recent characterization of (Mosheiff, Resch, Ron-Zewi, Silas, Wootters, 2019) of configurations of codewords that are contained in random linear codes.  Our upper bound follows from a refinement of the techniques of (Guruswami, H{\aa}stad, Sudan, Zuckerman, 2002) and strengthens a previous result of (Li, Wootters, 2018), which applied to list-decoding rather than average-radius list-decoding.
\end{abstract}

\thispagestyle{empty}
\setcounter{page}{0}
\newpage
\section{Introduction}\label{sec:intro}
In coding theory, one is interested in the combinatorial properties of sets $\cC \subseteq \F_q^n$.\footnote{Here and throughout the paper, $\F_q$ denotes the finite field with $q$ elements.  In this we work only consider linear codes, so we always assume that the alphabet is a finite field.}
Such a set $\cC$ is called a \emph{code} of \emph{length} $n$ over the \emph{alphabet} $\F_q$, and the elements $c \in \cC$ are called \emph{codewords}.

List-decoding, introduced by Elias and Wozencraft in the 1950's \cite{Elias57,Wozencraft58}, is such a combinatorial property. 
For $p \in [0,1]$ and integer $L \geq 1$,
we say that a code $\cC \subseteq \F_q^n$ is \emph{$(p,L)$-list-decodable} if, for all $z \in \F_q^n$, 
\[ |\inset{c \in \cC \suchthat \delta(c,z) \leq p } | < L, \]
where $\delta(x,y) = \frac{1}{n}\inabs{\inset{ i \suchthat x_i \neq y_i } }$ denotes relative Hamming distance.  That is, $\cC$ is list-decodable if not too many codewords of $\cC$ live in any small enough Hamming ball.
In this paper, we are interested in the trade-offs between $p$, $L$, and the rate of the code $\cC$.   The \emph{rate} $R$ of $\cC$ is defined as
$ R = \frac{\log_q|\cC|}{n}.$
The rate lies in the interval $[0,1]$, and larger is better.  

\paragraph{Variations of list-decoding.}
In this work, we consider standard list-decoding along with two variations.

The first variation is a strengthening of list-decoding known as \em average-radius list-decoding. \em 
A code $\cC$ is \emph{$(p,L)$-average-radius list-decodable} if for any set $\Lambda \subseteq \cC$ of size $L$ and $z \in \F_q^n$,
\[ \frac{1}{L} \sum_{c \in \Lambda} \delta(c,z) \geq p. \]
It is not hard to see that $(p,L)$-average-radius list-decodability implies $(p,L)$-list-decodability, and this stronger formulation has led to stronger lower bounds than are achievable otherwise~\cite{GuruswamiN14}.
In addition to stronger lower bounds, average-radius list-decoding---essentially replacing a maximum with an average in the definition of list-decoding---is a natural concept, and it has helped establish connections between list-decoding and compressed sensing~\cite{CheraghchiGV13}.

The second variation, known as \em list-recovery, \em is a version where the ``noise'' is replaced by uncertainty about each symbol of the received word $z$.  Formally, we say that a code $\cC$ is \emph{$(\ell, L)$-list-recoverable} if for any sets $S_1, \ldots, S_n \subseteq \F_q$ with $|S_i| \leq \ell$ for all $i$, 
\[ |\inset{c \in \cC \suchthat c_i \in S_i \,\forall i } | < L. \]
List-recovery was originally used as a stepping-stone to list-decoding and unique-decoding (e.g., \cite{GuruswamiI01, GuruswamiI02, GuruswamiI03, GuruswamiI04}) but it has since become a useful primitive in its own right, with applications beyond coding theory~\cite{INR10, NPR12, GNPRS13, HIOS15, DMOZ19}.

\paragraph{Pinning down the output list size.}
We are motivated by the problem of pinning down the output list size $L$ for (average-radius) list-decoding and for list-recovery.  For all three of these problems, given $q$ and $p$ (respectively, $q$ and $\ell$), there exists an \emph{optimal rate}, denoted $R^*$. Namely, $R^*$ is the largest rate so that, for any $\epsilon > 0$, there are $q$-ary codes of rate $R^*-\epsilon$ and arbitrarily large length, which are $(p,L)$-(average-weight)-list-decodable (resp. $(\ell,L)$-list-recoverable), for some $L(q,p,\epsilon)$ (resp. $L(q,\ell,\epsilon)$). Importantly, $L$ must not depend on the length of the code. The \em list-decoding capacity theorem \em gives the dependence of $R^*$ on $q$ and $p$ (resp. $q$ and $\ell$): $R^* = 1 - h_q(p)$ for both standard and average-radius list-decoding,~\cite{Elias91,ZyablovP81} and $R^* = 1 - \log_q(\ell)$ for list-recovery~(e.g., \cite{RudraW18}).

We are interested in the trade-off between the list size $L$, the parameters $p,q,\ell$ of the problem, and this gap $\eps$; we refer to $\eps$ as the \emph{gap to capacity.}
Pinning down the list size $L$ is an important problem.  For example, for many of the algorithmic applications within coding theory, the list size represents a bottleneck on the running time of an algorithm that must check each item in the list before pruning it down~\cite{GuruswamiI04, DvirL12, GuruswamiX12, GuruswamiX13, GuruswamiK16}.  For applications in pseudorandomness, for example to expanders or extractors, the list size corresponds to the expansion or to the amount of entropy in the input, respectively, and it is of interest to precisely pin down these quantities.

We make progress on pinning down the output list sizes for the case of \em random linear codes. \em  A random linear code is a uniformly random subspace of $\F_q^n$ of certain dimension.  The list-decodability of random linear codes has been well studied for many reasons~\cite{ZyablovP81,GHK11,CheraghchiGV13, Wootters13, RudraW14a, RudraW18, LiWootters}.  First, it is a natural mathematical question that studies the interplay between two fundamental notions in $\F_q^n$: subspaces and Hamming balls.  Second, there are constructions of codes which use random linear codes (and their list-decodability) as a building block~\cite{GuruswamiI04, GuruswamiR08b, HemenwayW15, HemenwayRW17}, and improvements in the parameters of random linear codes will lead to improvements in these constructions as well.  Third, random linear codes can be seen as one way to partially derandomize completely random codes; this is especially motivating in the binary (or fixed alphabet) case, where we do not know of any explicit constructions of optimally list-decodable codes, linear or otherwise.

\subsection{Contributions} \label{sec:contributions}
Our main results are improved bounds on the list size of random linear codes.
We defer the formal theorem statements until after we have set up notation, but we informally summarize our results here.  Below, we consider codes of \em rate \em $R^* - \eps$, where as above we use $R^*$ to denote best achievable rate for each particular problem.
\begin{enumerate}
\item[(1)] \textbf{Lower bound on the list size for list-recovery of random linear codes.}  We show that if a random linear code of rate $R^* - \eps$ is list-recoverable with high probability with input list sizes $\ell$ and output list size $L$, then we must have $L = \ell^{\Omega(1/\eps)}$.  This is in contrast to completely random codes, for which the output list size is $L = O(\ell/\eps)$ with high probability.

This gap between random linear codes and completely random codes demonstrates that in some sense zero-error list-recovery behaves more like \em erasure-list-decoding \em \cite{G03} than it does like list-decoding with errors.  Such a gap is present between general and linear codes in erasure list-decoding, but as we see below, there is no such gap for list-decoding from errors.

Our result extends to the setting of list-recovery with erasures as well.  The formal theorem statement and proof can be found in Section~\ref{sec:list_rec}.

\item[(2)] \textbf{Better lower bounds on the list size for list-decoding random linear codes.}  We show that if a $q$-ary random linear code of rate $R^* - \eps$ is list-decodable with high probability up to radius $p$ with an output list size of $L$, then we must have $L \ge \floor{\frac{h_q(p)}{\eps}+0.99}$.
By \cite{LiWootters}, this result is tight for list-decoding of binary random linear codes up to a small additive factor.
As an immediate corollary, $L \ge \floor{\frac{h_q(p)}{\eps}+0.99}$ for \emph{average-radius} list-decoding of random linear codes as well, and, as we will see below, this is also tight for binary random linear codes, up to a small additive factor.  We conjecture that the leading constant $h_q(p)$ is also correct for $q > 2$.

Previous work~\cite{GuruswamiN14} has established that $L = \Omega(1/\eps)$, but to the best of our knowledge this is the first work that pins down the leading constant.  In particular, \cite{GuruswamiN14} shows that, in the situation above, we have $L \geq c_{p,q}/\eps$, where $c_{p,q}$ is a constant that goes to zero as $p$ goes to $1 - 1/q$.  
In contrast, we show below that the leading constant is at least $h_q(p)$, which goes to $1$ as $p$ goes to $1 - 1/q$.

The formal theorem statement and proof can be found in Section~\ref{sec:list_dec}.

\item[(3)] \textbf{Completely pinning down the list size for average-radius list-decoding of binary random linear codes.}  We prove a new upper bound on the \em average-radius \em list-decodability of binary random linear codes, which matches our lower bound, even up to the leading constant.  More precisely, we show that with high probability, a random binary linear code of rate $R^* - \eps$ is \em average-radius \em list-decodable up to radius $p$ with $L \leq \floor{h_2(p) /\eps} + 2$.

Such a bound was known for standard list-decoding~\cite{LiWootters}, but our upper bound holds even for the stronger notion of average-radius list-decoding, and improves the additive constant by 1.\footnote{Under our definition of list-decoding, \cite{LiWootters} show $(p,L)$ list-decodability with $L=\floor{h(p)/\epsilon}+3$.}  In particular, this shows that for both list-decoding and average-radius list-decoding of binary random linear codes, the best possible $L$ is concentrated on at most three values: $\floor{h_q(p) /\eps} + 2,\floor{h_q(p) /\eps} + 1$ and $\floor{h(p)/\eps+0.99}$.
This tight concentration demonstrates the sharpness of our upper and lower bound techniques.

The formal theorem statement and proof can be found in Section~\ref{sec:upper_bd}.
\end{enumerate}

\subsection{Overview of techniques}\label{sec:techniques}
In this section, we give a brief overview of our techniques.

\paragraph{Lower bounds.}
To illustrate the techniques for our lower bounds, we warm up with a back-of-the-envelope calculation which suggests why the ``right'' answer for our result (1) above is $\ell^{\Omega(1/\eps)}$.  

Consider a random linear code $\cC \subset \F_q^n$, of rate $R = 1 - \log_q(\ell) - \eps$, where $\eps \in (0,\frac 12)$.
That is, $\cC$ is the kernel\footnote{This is one of several natural models for a random linear code. Another possible model is taking a uniformly random subspace of dimension $Rn$. It is not hard to see that the total variation distance between these distributions is exponentially small. In particular, our model yields a code of dimension exactly $Rn$ with probability $1-\exp{-\Omega(n)}$.} of a uniformly random matrix sampled from  $\F_q^{(1-R)n\times n}$. 
Suppose that $\ell$ is a prime power, and $q = \ell^t$ for some $t\ge 2$.  Thus, $\F_\ell$ is a sub-field of $\F_q$. Let $D$ be an integer slightly smaller than $\frac 1{2\eps}$ and let $L = \ell^D \ge \ell^{\Omega(1/\eps)}$. We claim that $\cC$ is unlikely to be $(\ell,L)$-list-recoverable.

Given a matrix $M\in \F_q^{n\times D}$, we write $M\subseteq \cC$ (``$\cC$ contains $M$'') to mean that each of the columns of $M$ is a codeword in $\cC$. 

Let $\cM$ denote the set of all full-rank matrices $M\in \F_q^{n\times D}$ that have the following property: for every row $M_i$ of $M$ there exists some $x_i\in \F_q^*$ such that all entries of $M_i$ belong to the set $x_i\cdot \F_\ell$. 

We will show that $\cM$ is \emph{bad} and \emph{abundant}. By \emph{bad} we mean that a linear code containing a matrix from $\cM$ cannot be $(\ell,L)$-list-recoverable. We say that $\cM$ is \emph{abundant} (for the rate $R$) if a random linear code of rate $R$ is likely to contain at least one matrix from $\cM$. Clearly, the combination of these properties means that $\cC$ is unlikely to be $(\ell,L)$-list-recoverable.

We first prove that $\cM$ is bad. Assume that $\cC$ contains some matrix $M\in \cM$. By linearity of the code, $\cC$ also contains every vector of the form $Mu$, for $u\in \F_q^D$. In particular, consider the set of vectors $B:=\left\{Mu \mid u\in \F_\ell^D\right\}\subseteq \cC$. Observe that $\cC$ cannot be $(\ell,L)$-list-recoverable, since $B$ is a ``bad list'' for list-recoverability with these parameters: First, since $M$ has full-rank, $B$ is of cardinality $\ell^D=L$. Now, given $i\in [n]$, we need to show that there exists a subset $S_i\subseteq \F_q$ with $|S_i| = \ell$, such that $v_i\in S_i$ for all $v\in B$. We take $S_i$ to be the set $x_i\cdot \F_\ell$, which contains all entries of the row $M_i$. For $j\in [D]$, write $M_{i,j} = x_i\cdot w_j$ ($w_j\in \F_\ell$), and let $v = Mu$ for some $u\in \F_\ell^D$. Then
$$v_i = \sum_{j=1}^D M_{i,j}u_j = x_i\cdot \left(\sum_{j=1}^D w_j\cdot u_j\right) \in x_i\cdot \F_\ell,$$
and we conclude that $\cM$ is bad.

Showing that $\cM$ is abundant is harder, and at this stage we only provide some intuition for this fact. Let us compute the expected number of matrices $M\in \cM$ that are contained in $\cC$. First, we estimate the cardinality of $\cM$. One may generate a matrix in $\cM$ by choosing each of its rows in an essentially independent fashion.\footnote{We say ``essentially'' since the resulting matrix might not have full rank, but this happens with negligibly small probability.}  Choosing a row amounts to choosing one of $\frac{q-1}{\ell-1}$ sets of the form $x\cdot \F_\ell$ ($x\in \F_q^*$) and then taking each entry to be an element of that set. Accounting for multiple counting of the all-zero row, the number of possible rows is thus $\frac{q-1}{\ell-1}\cdot (\ell^D-1)+1$, which we approximate as $q\cdot \ell^{D-1}$ .  Thus, $|\cM| \approx (q\cdot \ell^{D-1})^n$. Next, it is not hard to see that a random linear code contains a given matrix of rank $r$ with probability $q^{-(1-R)\cdot r\cdot n}$. Consequently, for $M\in \cM$, we have $\Pr[M\subseteq \cC] = q^{{-(1-R) D n}}$. Therefore,
\begin{align*}
\mathbb{E}\inabs{ \inset{ M \in \mathcal{M} \suchthat  M \subseteq \cC } } &=  |\cM| q^{-(1-R)Dn} \approx \inparen{q\cdot \ell^{D-1}}^n\cdot q^{-(1-R)Dn} \\
&= \left(\ell^{-1}\cdot q^{1-\eps D}\right)^n = \ell^{(-1+t(1-\eps D))n},
\end{align*}
where, the penultimate equality is due to substituting $1-\log_q(\ell)-\eps$ for $R$. Finally, since $t\ge 2$ and $D < \frac 1{2\eps}$, the right-hand side of the above is $\ell^{\Omega(n)}$.
 Thus, in expectation, $\cC$ contains many ``bad" lists for list-recovery. 

\vspace{.3cm}
Of course, this back-of-the-envelope calculation does \em not \em yield the result advertised above.  It might be the case that, even though the expected number of $M \in \mathcal{M}$ so that $M \subseteq \cC$ is large, the probability that such an $M$ exists is still small.  In fact, as \cite{MRRSW} shows, there are simple examples where this does happen. 
Thus, proving that $\cM$ is abundant requires more work.

A standard approach to show that $\cM$ is abundant would be via the second-moment method.  
Recently, \cite{MRRSW} gave a general theorem which encompasses second-moment calculations in this context.  In particular, they showed that there is essentially only one reason that a set $\cM$ might not be abundant: 
there exists some matrix $A \in \F_q^{D \times D'}$, such that the set $\{MA\mid M\in \cM\}$ is small. If this occurs, we say that $\mathcal{M}$ is \em implicitly rare.\em%
\footnote{The term ``implicitly rare'' is used by the first version of \cite{MRRSW}, available at \url{https://arxiv.org/abs/1909.06430v1}.}  They used this result to study the list-decodability of random Low-Density Parity-Check codes, but we can use their result to do our second moment calculation.
We show that our example of $\mathcal{M}$ above\footnote{More precisely, we study an example similar to this one; the example above was slightly tweaked to simplify the exposition for this back-of-the-envelope explanation.} is not implicitly rare, by showing that there is no such linear map $A$. This establishes that the back-of-the-envelope calculation is in fact correct.
Appealing to the machinery of \cite{MRRSW}, rather than applying the second moment method from scratch, allows us to get tighter constants with slightly less work, and gives a more principled approach to our lower bounds; indeed, our result (2) follows the same outline.

The intuition for our second result (2) is similar: we give an example of a class $\mathcal{M}$ which is \emph{bad} for list-decoding and \em abundant. \em 
We define $\cM$ as follows: Let 
$u\in \F_q^D$ be a random vector with independent $\Bernoulli_q(p)$ entries, namely, each entry is $0$ with probability $1-p$, and chosen uniformly from $\F_q^*$ with probability $p$. Let $x$ be uniformly sampled from $\F_q$. Let $\tau$ denote the distribution (over $\F_q^D$) of the random vector $u+x\cdot 1_D$. Finally, define $\cM$ to be the set of all matrices $M\in \F_q^{n\times D}$, such that a uniformly sampled row of $M$ has the distribution $\tau$.
As before, we show that $\cM$ is abundant by
showing that $\mathcal{M}$ is not implicitly rare and using the result of \cite{MRRSW}.

\paragraph{Upper bounds.}
Our argument for our upper bound result (3) closely follows that of~\cite{LiWootters}, which itself builds on the argument of~\cite{GuruswamiHSZ02}. The argument imagines building the random linear code one dimension at a time and uses a potential function to show that, so long as we do not add too many dimensions, no ball intersects the code too much. We now provide an informal overview of our approach, specifically comparing and contrasting it with the arguments of Guruswami, H{\aa}stad, Sudan and Zuckerman~\cite{GuruswamiHSZ02}; and Li and Wootters~\cite{LiWootters}. 

Let $R = 1-h(\rho)-\eps$ and put $k:=Rn$ (which we assume for exposition is an integer). Note that sampling a random linear code of rate $R$ is the same as sampling $b_1,\dots,b_k \in \F_2^n$ independently and uniformly at random and outputting $\mathrm{span}\{b_1,\dots,b_k\}$. Consider the ``intermediate'' codes $\cC_i = \mathrm{span}\{b_1,\dots,b_i\}$; \cite{LiWootters} (following \cite{GuruswamiHSZ02}) define a potential function $S_{\cC_i}$ and endeavor to show that $S_{\cC_i}$ does not grow too quickly. The work \cite{GuruswamiHSZ02} demonstrated that this holds in expectation; the work \cite{LiWootters} improved their argument to show that it holds with high probability. In both cases the potential function is such that it is easy to show that, so long as $S_{\cC}$ is $O(1)$, the code $\cC$ is list-decodable. 

The potential function in these works keeps track of the radius $p$ list-size at each vector $x \in \F_2^n$, that is, the cardinalities $|\{c \in \cC_i:\delta(x,c) \leq p\}|$ for $i=1,\dots,k$, and shows that so long as $i$ is not too large all these cardinalities remain at most $L$. For average-radius list-decoding, we instead keep track of a sort of ``weighted'' list size, where codewords that are very close $x$ are weighted more heavily. We can reuse much of the analysis from \cite{LiWootters} to demonstrate that on the $k$-th step the potential function is still bounded by a constant (in fact, it is at most $2$). The real novelty in our argument is a demonstration that, assuming this potential function is small, the code is indeed $(p,L)$-average-radius list-decodable.  This step is more involved than the argument in \cite{GuruswamiHSZ02,LiWootters} to establish $(p,L)$-list-decodability. 

\subsection{Related work} \label{sec:related-work}

We now highlight some related work. 
In what follows, $\eps$ is always the ``gap-to-capacity'', i.e., if the capacity for a particular problem is $R^*$, then the result concerns codes of rate $R^*-\eps$. 

\paragraph{Lower bounds for list sizes of arbitrary codes.} It is known that a typical (i.e., uniformly random) list-decodable code of rate $R^* - \eps$ has list size $L = \Theta(1/\eps)$, and a natural question to ask is whether \emph{every} code requires a list of size $L = \Omega(1/\eps)$. Blinovsky (\cite{Blinovsky86,blinovsky2005code}) 
showed that lists of size $\Omega_p(\log(1/\eps))$ are necessary for list-decoding a code of rate $R^* - \eps$. Later, Guruswami and Vadhan~\cite{GuruswamiV05} considered the high-noise regime where $p=1-1/q-\eta$ and showed that lists of size $\Omega_q(1/\eta^2)$ are necessary. Finally, Guruswami and Narayanan~\cite{GuruswamiN14} showed that for \emph{average-radius} list-decoding, the list size must be $\Omega_p(1/\sqrt{\eps})$.  

\paragraph{Existing lower bounds for random linear codes.}

For the special case of random linear codes, Guruswami and Narayanan~\cite{GuruswamiN14} showed that lists of size $c_{p,q}/\eps$ are necessary. The constant $c_{p,q}$ is not explicitly computed (and in fact relies on a constant from \cite{GHK11} which we discuss below), but one can deduce from the proof that if $p$ tends to $1-1/q$ then $c_{p,q}$ will tend to 0. 
Their lower bound follows from a second moment method argument, i.e., they consider a certain random variable $X$ whose positivity is equivalent to the failure of a random linear code to be list-decodable, and then show that $\Var{X} = o(\E{X})^2$. 
In this sense our approach is similar to theirs, because we rely on results from \cite{MRRSW} which themselves are proved using a second moment method.  However, we are able to get stronger results (in the sense that our leading constant does not decay as $p \to 1 - 1/q$, and moreover is optimal for binary codes).  One of the reasons may be the notion of ``implicit rareness'' from \cite{MRRSW}, which provides a useful characterization of the lists contained in a random linear code.

The work \cite{GuruswamiN14} also established lower bounds on list-decoding random linear codes from \em erasures. \em  While we do not discuss list-decoding from erasures in this work (except in the sense that erasure list-recovery is a generalization of list-decoding from erasures), this result is relevant to our work because \cite{GuruswamiN14} established an exponential lower bound of the form $L \geq \mathrm{exp}(\Omega(1/\eps))$, in contrast to the list size $O(1/\eps)$ that is attained by uniformly random codes.  Thus, our results suggest that (zero-error) list-recovery behaves more like list-decoding from erasures than from errors, at least with respect to the list size of random (linear) codes.

\paragraph{Existing upper bounds for random linear codes.}
We now turn our attention to upper bounds on list sizes for random linear codes. A long line of works~\cite{ZyablovP81,GuruswamiHSZ02, GHK11,CheraghchiGV13, Wootters13, RudraW14a, RudraW18, LiWootters} has studied this problem, and we highlight the most relevant results now. 
Zyablov and Pinsker~\cite{ZyablovP81} showed that random linear codes of rate $R^* - \eps$ have lists of size at most $q^{1/\eps}$.\footnote{For list-recovery with input lists of size $\ell$, the argument of \cite{ZyablovP81} shows that the list size is at most $q^{\ell/\eps}$. Furthermore, their results for list-decoding also apply to average-radius list-decoding.} 
Guruswami, H{\aa}stad, Sudan and Zuckerman~\cite{GuruswamiHSZ02} first showed the existence of capacity-achieving binary linear codes with lists of size $O(1/\eps)$.
Li and Wootters~\cite{LiWootters} revisited their techniques and showed that in fact random linear codes of rate $R^* - \eps$ have lists of size $O(1/\eps)$ with high probability; moreover they computed the constant coefficient in the big-Oh notation.  However, neither of these results apply to either average-radius list-decoding or to list-recovery.
As discussed above in Section~\ref{sec:techniques}, our new upper bound is the result of an improvement of the techniques of \cite{LiWootters}, which extends their result to average-radius list-decoding.

As for larger alphabets, Guruswami, H{\aa}stad and Kopparty~\cite{GHK11} showed that there exists a constant $C_{p,q}$ for which random linear codes are $(p,C_{p,q}/\eps)$-list-decodable with high probability. Unfortunately, if $p$ tends to $1-1/q$ then this constant tends to infinity. To address this, an ongoing line of works~\cite{CheraghchiGV13,Wootters13,RudraW14,RudraW18} has studied the list-decodability of random linear codes in the ``high-noise regime'' where $p$ is close to $1 - 1/q$; these results also apply to average-radius list-decodability.  These results imply that for binary random linear codes, when $p = 1 - 1/q - \Theta(\sqrt{\eps})$, random linear codes with rate $R^* - \eps$ are average-radius list-decodable with list sizes $O(1/\eps)$.  However, the constant hiding in the big-Oh is not correct (in particular, the authors do not see how to make it smaller than $2$).  Moreover, these results only hold in a particular parameter regime for $p$ and $\eps$, and degrade as the alphabet size grows.

As for list-recovery, a result by Rudra and Wootters~\cite{RudraW18} guarantees that random linear codes with rate $R^* - \eps$ over sufficiently large alphabets $\F_q$ have lists of sizes at most $(q\ell)^{O(\log(\ell)/\eps)}$.  To the best of our knowledge, no lower bounds were known.

\paragraph{Relevant results for other ensembles of codes.}

Lastly, we discuss some other results concerning other code ensembles. First of all, recent work of \cite{MRRSW} shows that a random code from Gallager's ensemble of LDPC codes~\cite{Gal62} achieves list-decoding capacity with high probability.  More generally, they show that random LDPC codes have similar combinatorial properties to random linear codes, including list-decoding, average-radius list-decoding, and list-recovery.  As part of their approach, they develop techniques to characterize the lists that appear in a random linear code with high probability, which we utilize for our work.

Finally, we note that there are no known explicit constructions of list-decodable codes of rate $R^* - \eps$ which achieve a list size even of $O(1/\eps)$.  
Over large alphabets,
the best 
explicit constructions of capacity-achieving list-decodable or list-recoverable codes have list sizes at least $(1/\eps)^{\Omega(1/\eps)}$ (e.g., \cite{KoppartyRSW18,KRRSS19}).
Further, if one insists on \emph{binary} codes, or even codes over alphabets of size independent of $\eps$, we do not know of \em any \em explicit constructions of list-decodable codes with rate approaching $R^*$.

\paragraph{Two-point concentration.} We showed that the optimal list size $L$ of a random linear code is concentrated on at most three values for both list-decoding and average-radius list-decoding: $\floor{h(p)/\eps} + 2, \floor{h(p)/\eps} + 1$, and, if the value is different, $\floor{h(p)/\eps+0.99}$.

In \cite[Theorem 2.5]{LiWootters}, it was also shown that the optimal list size of a \em completely \em random binary code is concentrated on two or three values for list-decoding.
This type of concentration is also well studied in graph theory, where it is known that in Erd\H{o}s-R\'enyi graphs, a number of graph parameters are concentrated on two values.
Examples include the clique number (size of the largest clique) \cite{matula1972employee,bollobas1976cliques}, the chromatic number \cite{luczak1991note, alon1997concentration, achlioptas2005two}, and the diameter \cite{riordan2010diameter}.

\subsection{Discussion and open problems}

In this work, we have made progress on pinning down the output list sizes for (average-radius) list-decoding and list-recovery for random linear codes.  Before we continue with the technical portion of the paper, we highlight some open questions and future directions.

\begin{itemize}
\item We showed that random linear codes of rate $1-h_q(p)-\eps$ are not $(p,L)$-list-decodable for $L\sim\frac{h_q(p)}{\epsilon}$. We conjecture this lower bound is tight, i.e. that random linear codes of rate $1-h_q(p)-\eps$ are $(p,L)$-(average-radius) list-decodable for $L=\frac{h_q(p)}{\epsilon}(1+o(1))$, where the $o(1)\to 0$ as $\eps\to 0$. Our Theorem~\ref{thm:main-avg-rad} (and earlier in \cite{LiWootters} for list-decoding) shows it is true for $q=2$, and we conjecture this is true for larger $q$.

\item Our results show that list-decoding and average-radius list-decoding have essentially the same output list sizes over binary alphabets, for random linear codes.  It would be interesting to extend this to larger alphabets, or even to more general families of codes.  This is especially interesting given that there is an exponential gap in the best known lower bounds (on the list-size for arbitrary codes) between list-decoding and average-radius list-decoding for general codes.

\item We have used different techniques for our upper and lower bounds.  However, we think it is an interesting direction to use the characterization of \cite{MRRSW}---which we used to prove our lower bounds---to prove upper bounds as well.  This would entail showing that every sufficiently bad list is implicitly rare. 

\item Finally, we note that our lower bounds for list-recovery rely on the field $\F_q$ being an extension field (that is, $q = p^t$ for some $t > 1$).  
It is an interesting question whether or not an exponential lower bound also holds over prime fields.  
We note that other lower bounds on list-decoding and list-recovery for Reed-Solomon codes also apply only to extension fields~\cite{GR05,BKR09}; perhaps all of these bounds taken together are evidence that better list-decodability may be possible in general over prime fields. 
\end{itemize}

\subsection{Organization}

In Section~\ref{sec:prelims} we set up notation and formally state the results of \cite{MRRSW} that we build on for our lower bounds.  In Section~\ref{sec:list_rec} we state and prove our lower bound on list-recovery of random linear codes.   In Section~\ref{sec:list_dec} we state and prove our lower bound on the list-decodability of random linear codes.   In Section~\ref{sec:upper_bd} we prove our upper bound on the list-decodability of random linear codes.

\section{Preliminaries}\label{sec:prelims}
In this section, we set notation and introduce the notions and results from \cite{MRRSW} that we need for our lower bounds.

\paragraph{Notation.}
Unless otherwise specified, all logarithms are base $2$.  We use the notation $\exp{x}$ to mean $e^x$. 
For an integer $a$, we define $[a]:=\{1,\dots,a\}$.
For a vector $x\in \mathbb{F}_q^A$ and $I\subset [A]$, we use $x_I \in \mathbb{F}_q^{|I|}$ to denote the vector $(x_i)_{i\in I}$ with coordinates from $I$ in increasing order.  We use $\mathbf{1}_D$ to denote the all ones vector of length $D$.
For vectors $v$ and $w$, let $\Delta_H(v,w)$ denote the Hamming distance between $v$ and $w$, i.e., the number of coordinates on which they disagree.

We use several notions from information theory.  
Define the $q$-ary entropy $h_q:[0,1]\to[0,1]$ by
\begin{align}
  h_q(x) \defeq x\log_q(q-1) - x\log_q x - (1-x)\log_q(1-x)
\label{eq:hq}
\end{align}
We assume $q=2$ if $q$ is omitted from the subscript.

For a random variable $X$ with domain $\mathcal{X}$, we use $H(X)$ to denote the \emph{entropy} of $X$:
\[
  H(X) = -\sum_{x\in\mathcal{X}}^{} \Pr_X(x)\log(\Pr_X(x)).
\]
For a probability distribution $\tau$, we may also use $H(\tau)$ to denote the entropy of a random variable with distribution $\tau$.

Let $X$ be a random variable supported on $\mathcal{X}$ and $Y$ be a random variable supported on $\mathcal{Y}$.
We define the \emph{conditional entropy} of $Y$ given $X$ as
\[
  H(Y|X) = - \sum_{x\in \mathcal{X},y\in \mathcal{Y}}^{} p(x,y)\log\frac{p(x,y)}{p(x)}.
\]
It is easy to check that $H(X)-H(X|Y) = H(Y)-H(Y|X)$ and we call this the \emph{mutual information} $I(X; Y)$: 
\[
  I(X; Y) = H(X) - H(X|Y) = H(Y)-H(Y|X)
\]
For random variables $X,Y,Z$, we define the \emph{conditional mutual information} $I(X; Y| Z)$ by
\[
I(X; Y|Z) = H(X|Z) - H(X|Y,Z) = H(Y|Z) - H(Y|X,Z)
\]
Conditional entropy, mutual information, and conditional mutual information satisfy the \em data processing inequality: \em for any function $f$ supported on the domain of $Y$, we have
\[
H(X|f(Y)) \ge H(X|Y)\quad\text{and}\quad 
I(X; Y) \ge I(X; f(Y)) \quad\text{and}\quad 
I(X; Y| Z)\ge I(X; f(Y)|Z).
\]
We also use \emph{Fano's inequality}, which states that if $X$ is a random variable supported on $\mathcal{X}$ and $Y$ is a random variable supported on $\mathcal{Y}$, and if $f:\mathcal{Y}\to \mathcal{X}$ is a function and $p_{err} = \Pr_{X,Y}[f(Y)\neq X]$
\[
H(X|Y)\le h(p_{err}) + p_{err}\cdot \log(|\mathcal{X}|-1)
\]

We define 
\[ 
H_q(X) \defeq \frac{H(X)}{\log q},\qquad I_q(X; Y) = \frac{I(X; Y)}{\log q}. 
\]
and similarly for conditional entropy and conditional mutual information.

For a distribution $\tau$ on $\F_q^L$ and a matrix $A \in \F_q^{L' \times L}$, we define the distribution $A\tau$ on $\F_q^{L'}$ in the natural way by
\[ \Pr_{A\tau}(x) = \sum_{\{y \in \F_q^L \suchthat Ay = x\}} \Pr_\tau(y), \]
namely, $A\tau$ is the distribution of the random vector $Ay$, where $y\sim \tau$. 

We have defined list-decoding, average-radius list-decoding, and list-recovery in the introduction.  We will in fact consider a more general version of list-recovery, which also tolerates erasures:
\begin{definition}[List-recovery from erasures]\label{def:ListRecoveryErasures}
  A code $C\subset \F_q^n$ is \emph{$(\alpha,\ell,L)$-list-recoverable from erasures} if the following holds.  
Let $S_1, \ldots, S_n \subset \F_q$ be lists so that $|S_i| \leq \ell$ for at least $\alpha n$ values of $i$.  Then 
\[ \inabs{ \inset{ c \in \cC \suchthat \forall i \in [n], c_i \in S_i } } < L. \]
We take $\alpha=1$ if it is omitted.
\end{definition} 

\paragraph{Tools from \cite{MRRSW}.}
As discussed in Section~\ref{sec:techniques}, for our lower bounds we use tools from the recent work \cite{MRRSW}.  We work with matrices $M \in \F_q^{n \times L}$ ($L\in \N$), where we view the columns of $M$ as potential codewords in $\cC$.  We use the notation ``$M \subseteq \cC$'' to mean that the columns of $M$ are all contained in $\cC$.

We group together sets of such matrices $M$ according to their row distribution.  

\begin{definition}[$\tau_M$, $\dim(\tau)$, $\cM_{n,\tau}$]
Given a matrix $M \in \F_q^{n\times L}$, the empirical row distribution defined by the rows of $M$ over $\F_q^{L}$ is called the type $\tau_M$ of $M$.  That is, $\tau_M$ is the distribution so that for $v \in \F_q^L$,
\[ \Pr_{\tau_M}(v) = \frac{ \inabs{ \inset{ i \suchthat \text{the $i$'th row of $M$ is equal to $v$} } } }{n}. \]
For a distribution $\tau$ on $\F_q^L$, we use $\dim(\tau)$ to refer to $\dim(\spn(\supp(\tau)))$. 
We use $\cM_{n,\tau}$ to refer to the set of all matrices in $\F_q^{n \times L}$ which have empirical row distribution $\tau$.
\end{definition}

\begin{remark}\label{rem:nitpick}
  We remark that for some distributions $\tau$ over $\F_q^L$, the set $\cM_{n,\tau}$ may be empty due to $n\cdot \Pr_\tau(v)$ not being an integer.
  For such $\tau$ we can define $\cM_{n, \tau}$ to consist of matrices $M$ with either $\floor{n\cdot\Pr_\tau(v)}$ or $\ceil{n\cdot \Pr_\tau(v)}$ copies of $v$.
  This has a negligible effect on the analysis as we always take $n$ to be sufficiently large compared to other parameters, so for clarity of exposition we ignore this technicality.
\end{remark}

Given $M\in \cM_{n,\tau}$, note that $\cM_{n,\tau}$ consists exactly of those matrices obtained by permuting the rows of $M$. 
In particular, since the \emph{random linear code} model is invariant to such permutations, all of the matrices in $\cM_{n,\tau}$ have the same probability of being contained in $\cC$. 

As discussed in Section \ref{sec:techniques}, we prove a lower bound by exhibiting a distributions $\tau$ over $\F_q^L$ such that the corresponding set $\cM_{n,\tau}$ is both \emph{bad} and \emph{abundant}. When $\cM_{n,\tau}$ satisfies these properties, we say that $\tau$ itself is, respectively, \emph{bad} and \emph{abundant}.

The work \cite{MRRSW} characterizes which distributions $\tau$ satisfy the \emph{abundance} property, namely, which classes $\cM_{n,\tau}$ are likely to have at least one of their elements appear (as a matrix) in a random linear code of a given rate.  To motivate the definition below, suppose that the distribution $\tau$ has low entropy: $H_q(\tau) < \gamma \cdot \dim(\tau)$ for some $\gamma \in (0,1)$.  This implies that the class $\cM_{n,\tau}$ is not too big: more precisely, it is not hard to see that  
$|\cM_{n,\tau}| \le q^{H_q(\tau)\cdot n}\le q^{\gamma \dim(\tau) n}$.
Using a calculation like we did in Section~\ref{sec:techniques}, we see that, since $\cM_{n,\tau}$ is not very large, it is unlikely for a random linear code of rate less than $1 -\gamma$  to contain a matrix from $\cM_{n,\tau}$. 

However, this is not the only reason that $\cM_{n,\tau}$ might be unlikely to appear in a random linear code.  As is shown in \cite{MRRSW}, it could also be because a random output of $\tau$, subject to some linear transformation (perhaps to a space of smaller dimension), has low entropy. We call such distributions \em implicitly rare: \em

\begin{definition}[$\gamma$-implicitly rare] \label{def:implicitly-rare}
We say that a distribution $\tau $over $\F_q^{L}$ is \em $\gamma$-implicitly rare \em if there exists a full-rank linear transformation $A: \F_q^{L} \to \F_{q}^{L'}$ where $L' \le L $ such that

\[H_q(A\tau) < \gamma \cdot \dim(A\tau)\] 

\end{definition}
Observe that by taking $A$ to be the identity map, we recover the case where $\tau$ itself has low entropy. Furthermore, note that every matrix in $\cM_{n,A\tau}$ has all of its columns contained in the column-span of some matrix in $\cM_{n,\tau}$. This implies that if no matrix in $\cM_{n,A\tau}$ lies in a code, then no matrix in $\cM_{n,\tau}$ lies in the code. Thus, abundance of the distribution $A\tau$ implies abundance of $\tau$. 

For an illustrative example of an implicitly rare distribution, we refer the reader to \cite[Example 2.5]{MRRSW}. Specifically, the example provides a case where for some full-rank matrix $A$, we have $H_q(A\tau)/\dim(A\tau) > H_q(\tau)/\dim(\tau)$.

Essentially, \cite{MRRSW} shows that a row distribution $\tau$ is likely to appear in a random linear code (namely, $\tau$ satisfies the abundance property) if and only if it is \emph{not} implicitly rare. 
The following theorem follows from Lemma 2.7 in \cite{MRRSW}.\footnote{This is also given as Theorem 2.2 in the first version of \cite{MRRSW}, available at \url{https://arxiv.org/abs/1909.06430v1}.}

\begin{theorem}[Follows from Lemma 2.7 in \cite{MRRSW}]\label{thm:technical}
Let $R \in (0,1)$ and fix $\eta > 0$.
Let $\tau$ be a $(1 - R - \eta)$-implicitly rare distribution over $\F_q^L$ ($L\in \N$), and let $\cC$ be a random linear code of rate $R$. 
Then 
\[\Pr[\exists M \in \cM_{n,\tau}: M \subseteq \cC] \le q^{-\eta n}\]
Conversely, suppose that $\tau$ is not $(1 - R + \eta)$-implicitly rare. Then
\[\Pr[\exists M \in \cM_{n,\tau}: M \subseteq \cC] \ge 1- n^{O_{L,q}(1)} \cdot q^{-\eta n} \ . \]
\end{theorem}

The first part of the theorem follows from a natural first-moment method argument, while the second part follows from the analogous second-moment argument. We emphasize that it is important that we allow arbitrary full-rank linear transformations $A:\F_q^L \to \F_q^{L'}$ in Definition~\ref{def:implicitly-rare}: if we only allowed $A$ to be the identity map, the second part of the theorem would be false.

\section{Lower bounds for list-recovery with erasures}\label{sec:list_rec}
 Our main result in this section is the following.

\begin{theorem}\label{thm:main_lr}
Fix $0 \le \rho < 1$. Fix a prime power $\ell \ge 2$ and an integer $t\ge 2$, and let $q=\ell^t$. Fix $0 < \varepsilon \le \frac{1-\rho}{20t}$ and let $L = \ell^{\ceil{\frac{1-\rho}{20\epsilon}}}$. 
For $n\in \N$, let $\cC \subseteq \F_q^n$ denote a random linear code of rate $R:=1-(\rho +(1-\rho)\log_q(\ell))-\varepsilon$. Then the probability of $\cC$ being $(1-\rho,\ell,L)$-erasure list-recoverable is at most $q^{-\Omega(n)}$.
\end{theorem}

\subsection{Proof of Theorem~\ref{thm:main_lr}}
We will prove Theorem~\ref{thm:main_lr} below, after we build up the necessary building blocks.  
As discussed in Sections \ref{sec:techniques} and \ref{sec:prelims}, to prove Theorem \ref{thm:main_lr} we seek a distribution $\tau$ that is both is \emph{bad} and \emph{abundant}.  That is, $\cC$ should likely contains some matrix from $\cM_{n,\tau}$, and the corresponding codewords should yield a counterexample to the list-recoverability of $\cC$. 
We will describe our choice of $\tau$ in Definition~\ref{def:badtau}; we will show that it is bad in Proposition~\ref{prop:badtau}; and finally we will show that it is not implicitly rare (and hence abundant by Theorem~\ref{thm:technical}) in Lemma~\ref{lem:lr}.

Our construction of the distribution $\tau$ follows similar lines to that in Section~\ref{sec:techniques}.

\begin{definition}[The bad distribution $\tau$ for list-recovery lower bounds]\label{def:badtau}
Fix $\rho,\ell, t,q$ as in Theorem~\ref{thm:main_lr}.  Let $D \geq t$ be a positive integer.
Let $\F_\ell$ be a subfield of $\F_q$, where $q = \ell^t$ and $t \geq 2$.  Let $\alpha_1, \ldots, \alpha_{(q-1)/(\ell-1)}$ be a set so that $\alpha_i \F_\ell^*$ are disjoint cosets of $\F_\ell^*$ partitioning $\F_q^*$.  
Let $L = \ell^D$.  
Let $G \in \F_\ell^{L \times D}$ be the matrix whose rows are all of the distinct elements of $\F_\ell^D$.

Let $\sigma$ be the distribution that with probability $1 - \rho$
returns $\alpha_i u$ for $(i,u)$ uniform in $\inset{1, \ldots, \frac{ q-1}{\ell-1}} \times \F_\ell^D$; 
and with probability $\rho$
returns a uniformly random element of $\F_q^D$.

Let $\tau$ be the distribution given by $Gv$ for $v \sim \sigma$.
\end{definition}

To motivate this construction, consider first the $\rho = 0$ case. 
Now consider a matrix $M \in \F_q^{n \times L}$ that has row distribution given by $\tau$.  
If we ignore the coefficients $\alpha_i$, the columns of $M$ span a $D$-dimensional subspace of $\F_\ell^n$.  In particular, they are \em bad, \em in the sense that each coordinate of these codewords are contained in a list of size $\ell$ (namely, $\F_\ell$).  Moreover, as soon as any $D$ linearly independent columns of $M$ are contained in $\cC$, all of the columns of $M$ are contained in $\cC$; this suggests that it's relatively likely (compared to, say, a random matrix in $\F_q^{L \times n}$) that $M \subseteq \cC$.  These properties don't change when we multiply by the coefficients $\alpha_i$: each coordinate is now contained in some list $\alpha_i \F_\ell$ rather than $\F_\ell$ (notice that the fact that the $\alpha_i$ are coset representatives means that all of these possible lists are disjoint, other than zero), and it's still just as likely that $M \subseteq \cC$.  However, by throwing these multiples $\alpha_i$ into the mix, we have increased the size of $\cM_{n,\tau}$, making $\tau$ more \em abundant. \em
In particular, note that, over all choices of $(i,u)$, the value $\alpha_i u$ is distinct except when $u=0$.
Thus, $\tau$ has entropy close to the entropy of the uniform distribution on $(i,u)$, so $H_q(\tau)\approx \log_q(q\ell^{D-1}) \approx D (\log_q(\ell) + \frac{1}{D})$.
Using a similar idea, we can estimate the entropy of $A\tau$ for all matrices $A$, showing that $\tau$ is \emph{not} $\log_q(\ell) + \frac{1}{10D}$ implicitly rare, implying that it is abundant.

To generalize to the $\rho > 0$ case, the construction essentially ``frees'' a $\rho$ fraction of the coordinates relative to the $\rho  =0$ case.  This further increases the size of $\cM_{n,\tau}$ (making $\tau$ even more abundant), while still maintaining the badness property for list-recovery with a $\rho$ fraction of erasures.

\begin{proposition}[$\tau$ is bad] \label{prop:badtau}
Let $\tau$ be as in Definition~\ref{def:badtau}.  Let $\cC \subseteq \F_q^n$, and let $M \in \cM_{n,\tau}$.  If $M \subseteq \cC$, then $\cC$ is not $(1 - \rho, \ell, L)$-list-recoverable.
\end{proposition}
\begin{proof}
Suppose that $M \subseteq \cC$.  
Let $w_1, w_2, \ldots, w_n \in \F_q^L$ be the rows of $M$.
It suffices to show that there are input lists $S_1, \ldots, S_n$ so that $w_j \in S_j^L$ for all $j \in [n]$, and so that for at least $(1-\rho)n$ values of $j \in [n]$, we have $|S_j| \leq \ell$.  
Recall that each row $w_j$ of $M$ is of the form $Gv_j$ where a $(1 - \rho)$ fraction of the $v_j$ are of the form $\alpha_{i_j} \cdot u_j$ for $(i_j,u_j) \in [(q-1)/(\ell-1)] \times \F_\ell^D$, and a $\rho$ fraction of the $v_j$ are arbitrary vectors in $\F_q^D$.%
\footnote{As per Remark~\ref{rem:nitpick}, we may ignore the rounding issue that $\rho n$ may not be an integer. This is without loss of generality, as we may replace $\tau$ with a very similar distribution so that a $\lceil\rho n\rceil$ fraction of the $v_j$ are arbitrary in $\F_q^D$,
and adjust all parameters by a term that is $o(1)$ as $n \to \infty$.}

In the first case, set $S_j = \alpha_{i_j} \cdot \F_\ell$.  
Because the elements of $G$ are all in $\F_\ell$, all the coordinates of $w_j = Gv_j = \alpha_{i_j} G u_j$ lie in $S_j$.  Moreover by definition $|S_j| \leq \ell$.
In the second case, set $S_j = \F_q$.  
By definition all the coordinates of $w_j \in \F_q^L$ lie in $S_j = \F_q$. 

This completes the proof.
\end{proof}

Next, we will show that $\tau$ is not implicitly rare, which will imply that $\tau$ is abundant.

\begin{lemma}[$\tau$ is abundant]\label{lem:lr}
Let $\tau$ be as in Definition \ref{def:badtau}.  
Then $\tau$
is \emph{not} $\inparen{\rho + (1-\rho)\log_{q}(\ell) + \frac{1-\rho}{10D}}$-implicitly rare.
\end{lemma}

The proof of Lemma~\ref{lem:lr} is in Section~\ref{sec:lemlr} below.
Before we prove Lemma~\ref{lem:lr}, we use it to prove Theorem~\ref{thm:main_lr}.

\begin{proof}[Proof of Theorem~\ref{thm:main_lr}, assuming Lemma~\ref{lem:lr}]
Let $0 < \eps \le \frac{1-\rho}{20t}$.  Let $\tau$ be as in Definition~\ref{def:badtau}, choosing $D = \left\lceil \frac{1-\rho}{20 \eps} \right\rceil$. By our choice of $\varepsilon$, we indeed have $D\ge t$. Lemma~\ref{lem:lr} shows that $\tau$ is \em not \em $(\rho + (1 -\rho)\log_q(\ell) + \frac{1-\rho}{10D})$-implicitly rare. By choice of $D$, we have $\frac{1-\rho}{10D} > \eps$. From Theorem~\ref{thm:technical} with $\eta = \frac{1-\rho}{10D} - \eps$, we see that for any sufficiently large $n$, a random code of rate
\[ \inparen{ 1 - \inparen{ \rho + (1 - \rho)\log_q(\ell) + \frac{1-\rho}{10D} } } + \eta = 1 - \inparen{ \rho + (1 - \rho) \log_q(\ell)} - \eps \]
contains $\ell^D$ codewords given by a matrix $M \in \cM_{n,\tau}$ 
with probability at least $1 - q^{\Omega(\eps n)}$.
By Proposition~\ref{prop:badtau}, if this occurs, then $\cC$ is not $(1 -\rho, \ell, L)$-list-recoverable.
\end{proof}

\subsection{Proof of Lemma~\ref{lem:lr}}\label{sec:lemlr}
In this section we prove Lemma~\ref{lem:lr}, which will complete the proof of Theorem~\ref{thm:main_lr}.
We first prove the following technical lemma, which roughly states that a distribution with few ``collisions'' has entropy close to the uniform distribution.

\begin{lemma}
\label{lem:collision}
    Let $Z$ be a finite set, and for $z\in Z$, let $n_z$ be a nonnegative integer. Suppose that $N_1 = \sum_{z\in Z} n_z$ and that $N_2 = \sum_{z\in Z}\binom{n_z}{2}$. Then the distribution that samples an element $z\in Z$ with probability $n_z/N_1$ has entropy at least
    \begin{align}
        \log(N_1) - \log\left(1 + \frac{2N_2}{N_1}\right).
    \end{align}
\end{lemma}
\begin{proof}
    The entropy is
    \begin{align}
        \sum_{z\in Z}\frac{n_z}{N_1} \cdot\log\left(\frac{N_1}{n_z}\right)
        &= \log(N_1) - \sum_{z\in Z} \frac{n_z}{N_1}\cdot \log(n_z)\nonumber\\
        &\ge \log(N_1) - \log\left(\sum_{z\in Z} \frac{n_z}{N_1}\cdot n_z\right) 
        = \log(N_1) - \log\left(\frac{N_1+2N_2}{N_1}\right),
    \end{align}
    as desired. In the inequality, we used Jensen's inequality and that $\log(x)$ is concave.
\end{proof}

Next, we prove Lemma~\ref{lem:lr}.  We prove it
first for $\rho=0$, and then use the $\rho=0$ case to prove the general statement.

\begin{proof}[Proof of Lemma \ref{lem:lr} for $\rho = 0$]

Fix a matrix $A\in\mathbb{F}_q^{L'\times L}$, and let $\tau$ be as in Definition~\ref{def:badtau}.  Recall that the distribution $A \tau$ is given by $A v$ for $v \sim \tau$.\footnote{Throughout this proof, the output of $\tau$ is treated as a column vector.}
Our goal is to show that, for all $A$, the distribution $A\tau$ supported on $\mathbb{F}_q^{L'}$ 
has large entropy.

Let $D'$ be the rank of $AG\in\mathbb{F}_q^{L'\times D}$.
  
First we show that $A\tau$ has dimension $D'$.
  By definition, $\alpha e_i \in \supp(\sigma)$ for all $\alpha \in \F_q$ and all $i \in [D]$, so $\supp(\tau)$ contains
$\{G\cdot \alpha \cdot e_{i} : \alpha \in \mathbb{F}_q, i \in [D]\}$, and thus $\spn_{\mathbb{F}_q}(\supp(\tau)) = G\cdot \mathbb{F}_q^D$.
Hence,
  \begin{align}
  \dim(\spn_{\mathbb{F}_q}(\supp(A\tau)))
  &= \dim(\spn_{\mathbb{F}_q}(A\cdot \supp(\tau))) \nonumber\\
  &= \dim(A\cdot \spn_{\mathbb{F}_q}(\supp(\tau))) \nonumber\\
  &= \dim(AG\cdot \mathbb{F}_q^D) \nonumber\\
  &= D'
  \end{align}
  as desired.

  Next we show that the entropy $H_q(A\tau)$ is at least $D'(\log_q\ell + \frac{1}{10D})$.
  It suffices to prove that \[H(A\tau) \geq \log\left(\ell^{D'}\cdot q^{0.1D'/D}\right).\]
Since $AG \in \F_q^{L' \times D}$ has rank $D'$, there exist $D'$ linearly independent rows whose span contains all the rows of $AG$.
Let $W \in \F_q^{D' \times D}$ be the submatrix of $AG$ obtained by keeping these rows.
Note that for this $W$, for all $v,v' \in \F_q^D$ we have
$Wv = Wv'$ if and only if $AGv = AGv'.$

Since $W$ has rank $D'$ and $D' \leq D$, there are $D'$ linearly independent columns of $W$.  Suppose without loss of generality that they are the first $D'$ columns of $W$.
  Thus, we may write $W = [W\ind{1}|W\ind{2}]$ where $W\ind{1} \in \F_q^{D' \times D'}$ is invertible and $W\ind{2} \in \F_q^{D'\times (D-D')}$.  For any $v \in \F_q^D$, we may write $v=\begin{bmatrix}
  v\ind{1}\\v\ind{2}
  \end{bmatrix}$
  where $v\ind{1}\in\mathbb{F}_q^{D'}, v\ind{2}\in\mathbb{F}_q^{D-D'}$.
  Then (recalling $\sigma$ from Definition \ref{def:badtau})
  \begin{align}
    H(A\tau) 
    &= H(W\sigma) \nonumber\\
    &=H_{v\ind{1}\sim \mathbb{F}_\ell^{D'}, v\ind{2}\sim \mathbb{F}_\ell^{D-D'}, i\sim[(q-1)/(\ell-1)]}(\alpha_i W\ind{1}v\ind{1}+\alpha_i W\ind{2}v\ind{2}) \nonumber\\
    &\ge H_{v\ind{1}\sim \mathbb{F}_\ell^{D'}, v\ind{2}\sim \mathbb{F}_\ell^{D-D'}, i\sim[(q-1)/(\ell-1)]}(\alpha_i W\ind{1}v\ind{1}+\alpha_i W\ind{2}v\ind{2}| W\ind{2}v\ind{2}) \nonumber\\
   &= \sum_{w \in \F_\ell^{D'}} \Pr_{v\ind{2} \sim \F_\ell^{D-D'}}[ W\ind{2}v\ind{2} = w ] H_{v\ind{1} \sim \F_\ell^{D'}, i \sim[(q-1)/(\ell-1)]}( \alpha_i W\ind{1}v\ind{1} + \alpha_i w ) \nonumber\\
 &\geq \min_{w \in \F_q^{D'}}  H_{v\ind{1} \sim \F_\ell^{D'}, i \sim[(q-1)/(\ell-1)]}( \alpha_i W\ind{1}v\ind{1} + \alpha_i w ), \nonumber
  \end{align}
where above we are using the notation $H_{x \sim X}$ to denote that the randomness in the definition of the entropy is over the choice of a uniformly random $x$ in $X$.
  Thus, it suffices to show, for any fixed vector $w$, we have
  \begin{equation}\label{eq:needtoshow}
    H_{v\ind{1}\sim \mathbb{F}_\ell^{D'}, i\sim[(q-1)/(\ell-1)]}(\alpha_i W\ind{1}v\ind{1}+\alpha_i w) 
    \ge \log\left(\ell^{D'}\cdot q^{D'/10D}\right)
  \end{equation}

  Before finishing the proof, we first give some intuition for the remaining details.
  First consider the case $w=0$.
  Note that, over all choices of $v\ind{1}$ and $i$, the vectors $\alpha_i v\ind{1}$ are all distinct, except the all 0s vector.
  Thus, as $W\ind{1}$ is invertible, the distribution of $\alpha_i W\ind{1}v\ind{1}+\alpha_i w$ is close to the uniform distribution on approximately $\ell^{D'}\cdot \frac{q-1}{\ell-1}\approx \ell^{D'-1}\cdot q$ vectors, so the entropy is at least roughly $\log(\ell^{D'-1}\cdot q)$; this turns out to be enough. 
  
  When $w\neq 0$, we do not have the same near-uniform distribution, but we do have the following useful property that carries over from the $w=0$ case: for a fixed $w\in \mathbb{F}_q^{D'}$ and $\alpha_i\neq \alpha_j$, there exists at most one pair $(v,u)\in(\mathbb{F}_\ell^D)^2$ such that $\alpha_i W\ind{1} v + \alpha_i w = \alpha_j W\ind{1} u + \alpha_j w$.
  To see this, suppose for contradiction there are two, $(v,u)$ and $(v',u')$.
  Then subtracting, we have $\alpha_i W\ind{1}(v-v') = \alpha_j W\ind{1}(u-u').$
  Since $W\ind{1}$ is invertible, we have $\alpha_i(v-v') = \alpha_j(u-u')\in \alpha_i\mathbb{F}_\ell^D\cap \alpha_j\mathbb{F}_\ell^D = \{0\}$.
  Thus, $v=v'$ and $u=u'$, a contradiction.
  Using this property, we know that, over the randomness of $v\ind{1}$ and $i$, there are not many ``collisions'' in $\alpha_iW\ind{1} v\ind{1} + \alpha_i w$, so the entropy should again be close to the entropy of the uniform distribution on $(i,v\ind{1})$, which is $\log(\ell^{D'-1}\cdot q)$.
 We can bound the entropy of such a distribution with few collisions with a careful application of Jensen's inequality (Lemma~\ref{lem:collision}).
 We then show the resulting bound is sufficient by some straightforward calculations.
  We note that our bounds hold for all prime powers $\ell$ and all $t\ge 2$, rather than simply for sufficiently large $\ell$ and $t$; this requires the argument to be a little more delicate. 
  
 We now show the rest of the proof. Fix $w \in \F_q^{D'}$.  
  For $z\in \mathbb{F}_q^{D'}$, let 
\[ n_z = \inabs{\inset{ (i,v\ind{1}) \suchthat v\ind{1}\in \mathbb{F}_\ell^{D'},\alpha_i W\ind{1} v\ind{1} + \alpha_i  w = z } }. \]
Thus, we have
\begin{equation}\label{eq:sumz}
\sum_{z}^{} n_z = \frac{q-1}{\ell-1}\cdot \ell^{D'} 
\end{equation}
since there are $(q-1)/(\ell-1)$ choices for $i$ and $\ell^{D'}$ choices for $v\ind{1}$.

Further, we have that  
\begin{equation}\label{eq:binomineq}
     \sum_{z}^{} \binom{n_z}{2}  \le  \binom{(q-1)/(\ell-1)}{2}.
\end{equation}
This is true because, on one hand, the left side counts the number of pairs $(i,v),(j,u)$ so that
\[ \alpha_iW\ind{1} u + \alpha_i w = \alpha_jW\ind{1} v + \alpha_j w, \] 
by caseworking on the value $z = \alpha_iW\ind{1} u + \alpha_i w = \alpha_jW\ind{1} v + \alpha_j w$.
On the other hand, for any fixed $i$ and $j$, there is at most one such pair $(i,v)$ and $(j,u)$, so the total number of pairs is at most $\binom{(q-1)/(\ell-1)}{2}$. 

For a uniform $i$ and $v\ind{1}\sim \mathbb{F}_q^{D'}$, the vector $\alpha_iW\ind{1}v\ind{1}+\alpha_i w$ equals a vector $z\in\mathbb{F}_q^{D'}$ with probability proportional to $n_z$.
Thus, by Lemma~\ref{lem:collision} with $N_1 = \frac{q-1}{\ell-1}\cdot \ell^{D'}$ and $N_2 = \binom{(q-1)/(\ell-1)}{2}$, we have 
  \begin{align}
    H_{v\ind{1}\sim \mathbb{F}_q^{D'}, i\sim[(q-1)/(\ell-1)]}(\alpha_i W\ind{1}v\ind{1}+\alpha_i w)
    \ &\ge \ \log\left(\frac{q-1}{\ell-1}\cdot \ell^{D'}\right) -  \log \left(1 + \frac{\left(\frac{q-1}{\ell-1}\right)\cdot \left(\frac{q-1}{\ell-1}-1\right)}{\frac{q-1}{\ell-1}\cdot \ell^{D'}}\right) \nonumber\\
    \ &= \   \log\left(\frac{q-1}{\ell-1}\cdot \ell^{D'}\right) -  \log \left(1 + \frac{q-\ell}{(\ell-1)\cdot \ell^{D'}}\right). \label{eq:progress}
  \end{align}

We now show that \eqref{eq:progress} implies \eqref{eq:needtoshow}.
Recall that $q = \ell^t$.  
We have
  \begin{align}
	1 + \frac{q-\ell}{(\ell-1)\cdot \ell^{D'}}
    \ &< \ 1+\frac{q-1}{(\ell-1)\ell^{D'}}  \nonumber\\
    \ &=  \ \frac{q-1}{(\ell-1)q^{0.1D'/D}}\cdot \left( \frac{(\ell-1)q^{0.1D'/D}}{q-1} +\left( \frac{q^{0.1/D}}{\ell} \right)^{D'}\right) \nonumber\\
    \ &\le \ \frac{q-1}{(\ell-1)q^{0.1D'/D}}\cdot \left( \frac{(\ell-1)q^{0.1}}{q-1} +\left( \frac{\ell^{0.1}}{\ell} \right)^{D'}\right) \nonumber\\
    \ &\le \ \frac{q-1}{(\ell-1)q^{0.1D'/D}}\cdot \left( \frac{(\ell-1)\ell^{0.2}}{\ell^2-1} + \frac{1}{\ell^{0.9}}\right) \nonumber\\
    \ &\le \ \frac{q-1}{(\ell-1)q^{0.1D'/D}}\cdot \left( 0.4 + 0.6\right). \nonumber
\end{align}
Therefore we conclude that
\begin{align}
	1 + \frac{q-\ell}{(\ell-1)\cdot \ell^{D'}}
    \ &< \ \frac{q-1}{(\ell-1)q^{0.1D'/D}}.
    \label{eq:moreprogress}
  \end{align}
  In the first inequality we used that $q-\ell < q-1$.
  In the second inequality, we used that $D'\le D$ and that $q^{1/D}\le q^{t/D}\le \ell$ (recall $t\le D$).
  In the third inequality, we used that $\frac{x^{0.1}}{x-1}$ is decreasing for $x>1$, that $\ell^2\le q$, and that $D'\ge 1$.
  In the fourth inequality we used $\ell\ge 2$ and that $\frac{(x-1)x^{0.2}}{x^2-1}$ and $x^{-0.9}$ are decreasing for $x\ge 2$.

Combining \eqref{eq:moreprogress} with \eqref{eq:progress} proves \eqref{eq:needtoshow}, so
    \[H(A\tau)> \log(q^{0.1D'/D}\ell^{D'})\]
    which gives us
    \[H_q(A\tau)> \log_q(q^{0.1D'/D}\ell^{D'}) = D'\left( \log_q(\ell)+\frac{1}{10D} \right)\]
  These computations holds for any matrix $A$ of rank $L' \le L$, so we have that $\tau$ is \emph{not} $\inparen{\log_q(\ell)+\frac{1}{10D}}$-implicitly rare.
\end{proof}

This concludes the proof when $\rho = 0$; we continue to the case when $\rho > 0$.

\begin{proof}[Proof of Lemma \ref{lem:lr} for $\rho > 0$]
We need to show that for any $A\in\mathbb{F}_q^{L'\times L}$ such that $AG\in\mathbb{F}_q^{L'\times D}$ has rank $D'\le D$, the entropy of $H_q(A\tau)$ is at least $D'(\rho + (1-\rho)\log_q\ell + \frac{1-\rho}{10D})$. To see this first note that 
  
  \[\tau = \begin{cases} G\alpha_iv, \quad &
\text{ with probability } 1 -\rho,
\text{where }i \sim [1,\dots,(q-1)/(\ell-1)], v \sim\mathbb{F}_\ell^D \\
  				 Gw, \quad  &
\text{ with probability } \rho,
\text{where }w\sim\mathbb{F}_q^D 
\end{cases}\]

Define $\tau_1$ as the distribution of $G\alpha_iv$ and $\tau_2$ as the distribution of $Gw$. Now note that $\tau = (1-\rho)\tau_1 + \rho\tau_2$. 
Let $g(x) = -x\log x$.
Since $g(x)$ is concave, we have
\begin{align}
    H_q(A\tau)
    &=\sum_{v\in \mathbb{F}_q^L}g(\Pr_{A\tau}(v)) \nonumber\\
    &=\sum_{v\in \mathbb{F}_q^L}g\big((1-\rho)\cdot \Pr_{A\tau_1}(v)+\rho\cdot \Pr_{A\tau_2}(v)\big) \nonumber\\
    &\ge\sum_{v\in \mathbb{F}_q^L}(1-\rho)\cdot g(\Pr_{A\tau_1}(v))+\rho\cdot g(\Pr_{A\tau_2}(v)) \nonumber\\
    &= (1-\rho)\cdot H_q(A\tau_1)+ \rho\cdot  H_q(A\tau_2).
\end{align}

From the $\rho = 0$ case we already know that  
\[ H_q(A\tau_1) \ge D'\left( \log_q(\ell)+\frac{1}{10D} \right)\]
When $w$ is uniform on $\mathbb{F}_q^D$, then $AGw$ is uniformly distributed in the $\mathbb{F}_q$-span of $AG$ which has rank $D'$, so $H_q(A\tau_2) = D'$.
We thus have 
\begin{align}
H_q(A\tau) &\ge (1-\rho)H_q(A\tau_1)+ \rho H_q(A\tau_2)
\ge (1-\rho)\cdot D'\left(\log_q(\ell) + \frac{1}{10D}\right) + \rho D',
\end{align}
as desired.
\end{proof}

\section{Lower bounds for list-decoding with errors}\label{sec:list_dec}
Our main theorem in this section is the following.

\begin{theorem}\label{thm:main_ld}
Fix a prime power $q$, fix $p \in (0,1-\frac{1}{q})$, and fix $\delta \in(0,1)$.
There exists $\varepsilon_{p,q,\delta}>0$ such that for all $\varepsilon\in(0,\varepsilon_{p,q,\delta})$ and $n$ sufficiently large, a random linear code in $\mathbb{F}_q^n$ of rate $1 - h_q(p) - \eps$ is \emph{not} $\left(p, \floor{\frac{h_q(p)}{\eps}-\delta}\right)$-list-decodable with probability $1-q^{-\Omega(n)}$.
\end{theorem}

\subsection{Proof of Theorem~\ref{thm:main_ld}}

Our proof of Theorem~\ref{thm:main_ld} below follows the same outline as the proof of Theorem~\ref{thm:main_lr} above.
We first define a bad distribution $\tau$ in Definition~\ref{def:badtau_ld}; then we will show that it is bad in Proposition~\ref{prop:badtau_ld}; then we will show that it is not implicitly rare (and hence abundant by Theorem~\ref{thm:technical}) in Lemma~\ref{lem:ld}.  Finally we will prove Theorem~\ref{thm:main_ld} from these pieces.

  Below, we let $\Ber_q(p)$ be the distribution that returns $0\in\mathbb{F}_q$ with probability $1-p$ and any other element of $\mathbb{F}_q$ with probability $\frac{p}{q-1}$.

\begin{definition}[The bad distribution $\tau$ for list-decoding lower bounds]\label{def:badtau_ld}
  Let $p\in(0,1-\frac{1}{q})$ and $\delta > 0$.  Choose $L > 0$.  Define the distribution $\tau$ on $\F_q^L$ 
as the distribution of the random vector $u+\alpha \mathbf{1}_L$, where $u\sim \Ber_q(p)^L$, and $\alpha$ is sampled independently and uniformly from $\F_q$.
\end{definition}

First, we observe that $\tau$ is indeed bad, in the sense that it provides a counter-example to list-decodability.

\begin{proposition}[$\tau$ is bad]\label{prop:badtau_ld}
Let $\tau$ be as in Definition~\ref{def:badtau_ld}.  Let $\cC \subseteq \F_q^n$ and let $M \in \cM_{n,\tau}$.  If $M \subseteq \cC$, then $\cC$ is not $(p, L)$-list-decodable.
\end{proposition}

\begin{proof}
Let $M \in \cM_{n,\tau}$.  We want to show that 
the columns of $M$ all lie in a single ball of radius $pn$.

By definition of $\tau$ and $\cM_{n,\tau}$, we may write the $j$-th row of $M$ as $u\ind{j}+\alpha_j\mathbf{1}_L$, so that the empirical distribution of the pairs $(u\ind{j},\alpha_j)_{1\le j\le n}$ is $\Ber_q(p)^L\times \Uni(\F_q)$.\footnote{This is without loss of generality: if not, as per Remark \ref{rem:nitpick}, we can associate pairs with rows so that the empirical distribution is close to $\Ber_q(p)^L\times \Uni(\F_q)$ up to an additive factors that are $o(1)$ as $n\to\infty$.  After adjusting parameters, this has a negligible effect on the analysis and final result.}

For any $i\in[L]$, the number of $j\in[n]$ such that $M_{i,j}=u\ind{j}_i + \alpha_j\neq \alpha_j$ is exactly the number of times $u\ind{j}_i\neq 0$, which is $pn$, since $u\ind{j}_i$ is distributed as $\Ber_q(p)$. 
Thus, each column $M_{i,*}$ of $M$ has distance at most $pn$ from the word $(\alpha_1,\dots,\alpha_n)$,
so that any code containing $M$ has $L$ codewords in a ball of radius $pn$ and hence is not $(p,L)$-list-decodable.
\end{proof}

Next, we show that $\tau$ is appropriately implicitly rare for large enough $L$.

\begin{lemma}\label{lem:ld} 
  Let $p\in(0,1-\frac{1}{q})$ and let $\delta > 0$.
  There exists $L_{p,q,\delta}$ such that, for $L\ge L_{p,q,\delta}$, the distribution $\tau$
given in Definition~\ref{def:badtau_ld}
is \emph{not} $\inparen{h_q(p) + \frac{h_q(p)}{L+\delta}}$-implicitly rare.
\end{lemma}

We prove Lemma~\ref{lem:ld} in Section~\ref{sec:pflemld} below.
Before we prove Lemma~\ref{lem:ld}, we show how to use it to prove Theorem~\ref{thm:main_ld}.

\begin{proof}[Proof of Theorem \ref{thm:main_ld}, assuming Lemma~\ref{lem:ld}]
Let $L_{p,q,\delta/2}$ be as in Lemma~\ref{lem:ld} and choose $\varepsilon_{p,q,\delta} \defeq \frac{h_q(p)}{L_{p,q,\delta/2}+1}$.
Fix $\varepsilon < \varepsilon_{p,q,\delta}$.
Let $L=\floor{\frac{h_q(p)}{\varepsilon}-\delta}$.
Let $\tau$ be as in Definition \ref{def:badtau_ld} with this choice of $L$.
By Lemma~\ref{lem:ld}, as $L\ge L_{p,q,\delta/2}$, $\tau$ is not $\inparen{h_q(p) + \frac{h_q(p)}{L+\delta/2}}$-implicitly rare.  Thus, as $\varepsilon \le \frac{h_q(p)}{L+\delta} < \frac{h_q(p)}{L+\delta/2}$, there is some constant $c_{p,q,\eps}> 0$ so that $\tau$ is not $(h_q(p) + \varepsilon+c_{p,q,\varepsilon})$-implicitly rare.

Then Theorem~\ref{thm:technical} with $\eta=c_{p,q,\varepsilon}$ tells us that, for $n$ sufficiently large, a random linear code of rate $1 - (h_q(p) + \eps+c_{p,q,\varepsilon}) + c_{p,q,\varepsilon} = 1-h_q(p)-\varepsilon$ contains $L$ codewords given by some matrix $M \in \cM_{n,\tau}$ with probability at least $1 - q^{-\Omega_{p,q,\eps}(n)}$.

Finally, Proposition~\ref{prop:badtau_ld} implies that $\cC$ is not $(p,L)$-list-decodable.  Our choice of $L$ proves the theorem.
\end{proof}

\subsection{Proof of Lemma~\ref{lem:ld}}\label{sec:pflemld}
In this section we prove Lemma~\ref{lem:ld}, which completes the proof of Theorem~\ref{thm:main_ld}.
To prove Lemma~\ref{lem:ld} we need to prove that $A\tau$ has high entropy for any matrix $A$.
We begin with the following lemma, which essentially shows that this is true when $A$ is either the $L\times L$ identity $I_L$ or an $L\times (L+1)$ matrix with the identity and an additional column with all nonzero entries.
\begin{lemma}
  \label{lem:ber-p}
  Let $q$ be a prime power, $p\in\inparen{0,1-\frac{1}{q}}$, $p'\in\left[p,1-\frac{1}{q}\right]$, and $\delta>0$.
  There exists $L_{p,q,\delta}$ such that, for all $L\ge L_{p,q,\delta}$ and $0\le d\le L$, the following holds.
  Let $w$ be a fixed vector in $\mathbb{F}_q^d$ all of whose entries are nonzero.
  Let $v$ be a vector sampled from $\Ber_q(p)^{d}$ and let $\alpha$ be sampled from $\Ber_q(p')$.
  Then
  \begin{align}
    \label{eq:ber-p}
    H_q(v + \alpha w) \ge d\cdot \left(h_q(p) + \frac{h_q(p)}{L+\delta}\right).
  \end{align}
\end{lemma}

\begin{proof}
  If $d=0$, the assertion is trivial, so assume $d\ge1$. As a guide to the reader, we emphasize that throughout the proof the vector $v$ and the field element $\alpha$ are random variables, while the vector $w$ is fixed.

We will bound $H_q(v + \alpha w)$ in two cases, one when $d$ is small (relative to $L$) and one when $d$ is large.  (The precise definitions of ``small'' and ``large'' will be determined below.)

First we consider the case where $d$ is small.  We have (for any $d$) that
  \begin{align}
    H_q(v + \alpha w)
    \ &= \   H_q(v_1+\alpha w_1,v_2+\alpha w_2,\dots,v_d+\alpha w_d) \nonumber\\
    \ &= \   H_q(v_2+\alpha w_2,\dots,v_d+\alpha w_d|v_1+\alpha w_1) + H_q(v_1+\alpha w_1) \nonumber\\
    \ &\ge \ H_q(v_2+\alpha w_2,\dots,v_d+\alpha w_d|v_1,\alpha) + H_q(v_1+\alpha w_1)  \nonumber\\
    \ &= \ H_q(v_2,\dots,v_d) + H_q(v_1+\alpha w_1)  \label{eq:p1}
\end{align}
  The second equality uses the definition of conditional entropy.
  The inequality follows from the data processing inequality. The last equality uses the fact that $w$ is a fixed vector so once $\alpha$ is known, $\alpha w_2,\dots,\alpha w_d$ are also known, along with the assumption that the $v_1$ is independent of $v_2,\dots,v_L$.

Now, $v_1+\alpha w_1$ is nonzero if $v_1=0$ and $\alpha\neq 0$, if $v_1\neq 0$ and $\alpha = 0$, or if $v_1,\alpha\neq 0$ and $v_1+\alpha w_1\neq 0$. 
This happens with probability $p^* = (1-p')p+(1-p)p'+\frac{(q-2)pp'}{q-1}$.  In the case that $v_1+\alpha w_1$ is nonzero, then by symmetry each nonzero element of $\mathbb{F}_q$ has equal probability.  Thus $v_1 + \alpha w_1$ it is distributed as $\Ber_q(p^*)$.  One can check that $H_q(\Ber_q(p^*)) = h_q(p^*)$, so from \eqref{eq:p1} we have
\begin{align}
H_q(v + \alpha w) 
&\ge H_q( v_2, \dots, v_d) + h_q(p^*) \\
&= (d-1)\cdot h_q(p) + h_q(p^*). \label{eq:p2}
\end{align}

Since 
\[ p<2p(1-p) + \frac{(q-2)p^2}{q-1}\le p^*\le 1-\frac{1}{q}, \]
and $h_q(\cdot)$ is strictly increasing on $(0,1-\frac{1}{q})$, we have
$ h_q(p^*) \ge h_q(p)+\varepsilon_{p,q} $
for some $\varepsilon_{p,q} > 0$ depending only on $p$ and $q$. The first inequality uses the assumption $p < 1-1/q$ while the second inequality follows from the fact that $p^*$ increases with $p'$ and $p' \geq p$. 
  Thus, when $d\le \varepsilon_{p,q}\cdot L$, \eqref{eq:p2} implies that
  \begin{align}
    H_q(v + \alpha w)\ge d\cdot h_q(p) + \varepsilon_{p,q} 
    \ge d \cdot \left(h_q(p) + \frac{1}{L}\right)
    > d \cdot \left(h_q(p) + \frac{h_q(p)}{L+\delta}\right),
  \end{align}
where in the last inequality we have used that $\delta > 0$ and $h_q(p) < 1$. This lower bounds $H_q(v + \alpha w)$ in the case when $d$ is ``small,'' specifically when $d < \eps_{p,q} \cdot L$.

Next we handle the case when $d$ is ``large.''
We have (for any $d$) that
  \begin{align}
    H_q(v + \alpha w)
    \ &= \   H_q(v+\alpha w | \alpha)
    + H_q(\alpha) - H_q(\alpha | v + \alpha w) \nonumber\\
    \ &= \  d\cdot h_q(p) + h_q(p') - H_q(\alpha | v +\alpha w) \\
    \ &\ge \ d \cdot h_q(p) + h_q(p) - H_q(\alpha |v+\alpha w).
  \end{align}
  It thus suffices to show that $H_q(\alpha|v+\alpha w)$ is ``small''.
  To do this, we leverage Fano's inequality.

  Let $\hat\alpha$ be the element of $\mathbb{F}_q$ that minimizes the Hamming distance $\Delta_H(\hat\alpha w,v+\alpha w)$, breaking ties arbitrarily.
  In expectation a $1-p > \frac{1}{q}$ fraction of the $d$ coordinates of $v$ are 0.
Similarly, for any vector $w'\in\mathbb{F}_q^{d}$ with all nonzero entries, in expectation a $\frac{p}{q-1}<\frac{1}{q}$ fraction of the coordinates of $v$ agree with $w'$.

By Hoeffding's inequality, for any nonzero $\zeta \in \F_q$, 
\begin{align}
\PR{ \Delta_H(v, \zeta w) \ge \frac{d}{q} } \leq 2\mathrm{exp}\inparen{ -2d \inparen{ \frac{1}{q} - \frac{p}{q-1}} } = \mathrm{exp}(-\Omega_{p,q}(d))
\label{eq:hoef-1}
\end{align}
and similarly
\begin{align} 
\PR{ \Delta_H(v, \mathbf{0}) \le \frac{d}{q} } \leq 2 \mathrm{exp} \inparen{ -2d \inparen{ { 1-p-\frac{1}{q} } } } = \mathrm{exp}(-\Omega_{p,q}(d)). 
\label{eq:hoef-2}
\end{align}
If none of the events in \eqref{eq:hoef-1} and \eqref{eq:hoef-2} hold, then we have $\Delta_H(\alpha' w, v+\alpha w) < d/q$ for all $\alpha'\neq \alpha$ and $\Delta_H(\alpha w,v+\alpha w)> d/q$, in which case $\hat\alpha = \alpha$. Thus, by the union bound over all $q$ events in \eqref{eq:hoef-1} and \eqref{eq:hoef-2}, the probability that $\alpha\neq \hat\alpha$ is at most
\[ p_{err} \defeq \PR{ \hat\alpha\neq \alpha} \leq q\cdot \mathrm{exp}\inparen{ -\Omega_{p,q}(d) } = \mathrm{exp}\inparen{ -\Omega_{p,q}(d) }. \]
  By Fano's inequality, as $\alpha$ takes at most $q$ values and as $\hat \alpha$ is a function only of $v+\alpha w$, we have
  \begin{align}
    H_q(\alpha|v+\alpha w) 
    = \frac{1}{\log q}H(\alpha|v+\alpha w) 
    \le \frac{1}{\log q}(h(p_{err}) + p_{err}\cdot \log(q-1))
    \le \mathrm{exp}\inparen{-\Omega_{p,q}(d)}.
  \end{align}
  Thus, there exists some $d_{p,q,\delta}$ such that, for $d\ge d_{p,q,\delta}$, we have $H_q(\alpha|v+\alpha w)  \le \frac{\delta h_q(p)}{d+\delta}$, in which case
  \begin{align}
    H_q(v + \alpha w)
    \ &\ge \  d\cdot h_q(p) + h_q(p) - H_q(\alpha | v +\alpha w) \nonumber\\
    \ &\ge \ d\cdot h_q(p) + \frac{d}{d+\delta}\cdot h_q(p) \nonumber\\
    \ &\ge \ d\cdot \left(h_q(p) + \frac{h_q(p)}{L+\delta}\right). 
  \end{align}
This completes the case where $d$ is ``large.''

  We have shown that \eqref{eq:ber-p} holds when $d\le \varepsilon_{p,q}\cdot L$ and when $d\ge d_{p,q,\delta}$.
  Thus, for 
$$L\ge d_{p,q,\delta}/\varepsilon_{p,q}\defeq L_{p,q,\delta},$$ 
we have that \eqref{eq:ber-p} always holds, as desired.
\end{proof}

Using Lemma~\ref{lem:ber-p}, we may now prove Lemma~\ref{lem:ld} which says that $\tau$ is not implicitly rare.
\begin{proof} [Proof of Lemma \ref{lem:ld}]
Let $L_{p,q,\delta}$ be as in Lemma~\ref{lem:ber-p}.
Let $L\ge L_{p,q,\delta}$, and let $\tau$ be the corresponding distribution in the lemma statement.\footnote{As in Lemma \ref{lem:lr}, we treat the output of $\tau$ as a column vector.}
Fix a full-rank matrix $A$ of rank $L'$.
As $\tau$ is supported on $\mathbb{F}_q^L$, the rank of $A\tau$ is $L'$.
We show that $H_q(A\tau)\ge L'\cdot (h_q(p) + \frac{h_q(p)}{L+\delta})$.
At a high level, our strategy is to decompose the distribution $A\tau$ into several distributions that each have the set up of Lemma~\ref{lem:ber-p}. Furthermore, this decomposition has enough conditional independence that the entropy of $A\tau$ can be lower bounded by the sum of the entropies of the smaller distributions, which we can lower bound by Lemma~\ref{lem:ber-p}.

As $A$ is full-rank it must have exactly $L'$ rows.
Since permuting the coordinates of $\tau$ yields the same distribution $\tau$, permuting the columns of $A$ does not change the entropy $H_q(A\tau)$; thus, we may assume that the first $L'$ columns are linearly independent.
Furthermore, if $B$ is invertible, $H_q(BA\tau)=H_q(A\tau)$. 
Thus, by running Gaussian elimination on the rows of $A$, we may assume without loss of generality that
\begin{align}
  A = 
  \left[
\begin{array}{ccc|cccc}
 &  &  & \vert & \vert & \cdots & \vert \\
 & I_{L'} &  & w\ind1 & w\ind{2} & \cdots & w\ind{k}\\
 &  &  & \vert & \vert & \cdots & \vert \\
\end{array}
\right]
\end{align}
where $w\ind{1},\dots,w\ind{k}\in \mathbb{F}_q^{L'}$ and $k=L-L'$.
Let a sample from $\tau$ be given by 
\begin{align}
  \begin{bmatrix}
    v_1 \\
    \vdots\\
    v_{L'}\\
    \alpha_1\\
    \vdots\\
    \alpha_k
  \end{bmatrix}
  + 
  \begin{bmatrix}
    1 \\
     \\
    \vdots\\
    \\
    \\
    1
  \end{bmatrix}
  \cdot \alpha_{k+1}.
\end{align}
where $v_1,\dots,v_{L'},\alpha_1,\dots,\alpha_k\sim \Ber_q(p)$ and $\alpha_{k+1}\sim \Ber_q(1-\frac{1}{q})$.
(Note that this means that $\alpha_{k+1}$ is uniform on $\mathbb{F}_q$.)
Then $A\tau$ is given by
\begin{align}
  \label{eq:atau}
  \begin{bmatrix}
  v_1\\
  \vdots\\
  v_{L'}
  \end{bmatrix}
  +
  \sum_{i=1}^{k+1}
  \begin{bmatrix}
    \vert\\
    w\ind{i}\\
    \vert 
  \end{bmatrix}
  \cdot \alpha_i 
\end{align}
where we let $w\ind{k+1}$ be the product $A\cdot \mathbf{1}_{L}\in \mathbb{F}_q^{L'}$.
We emphasize that $v_1,\dots,v_{L'},\alpha_1,\dots,\alpha_{k+1}$ are independent random variables, while $A$ and $w\ind{1},\dots,w\ind{k}$ are fixed.

By definition of $A$ and $w\ind{k+1}$, 
for any coordinate $i \not\in \bigcup_{j=1}^k \supp(w\ind{j})$, we have $i \in \supp(w\ind{k+1})$.
Thus, $\bigcup_{i=1}^{k+1}\supp(w\ind{i}) = [L']$.
For $i=1,\dots,k+1$, let $J_i=\supp(w\ind{i})\setminus (\bigcup_{j=1}^{i-1}J_j)$  (when $i=1$ the union is the empty set), so that $J_1,\dots,J_{k+1}$ form a partition of $[L']$.
Recall that the notation $v_J\in \mathbb{F}_q^{|J|}$ denotes the vector $(v_i)_{i\in J}$ with coordinates from $J$ in increasing order.
We have
\begin{align}
  H_q(A\tau)
  \ &= \ H_q(A\tau | v_{J_{k+1}}, \alpha_{k+1}) + I_q(A\tau; v_{J_{k+1}}, \alpha_{k+1}) \nonumber\\ 
  \ &= \ H_q(A\tau | v_{J_k},v_{J_{k+1}}, \alpha_k,\alpha_{k+1}) + I_q(A\tau; v_{J_k}, \alpha_k | v_{J_{k+1}},\alpha_{k+1}) + I_q(A\tau; v_{J_{k+1}}, \alpha_{k+1}) 
\end{align}
Continuing, we have 
\begin{align}
  \label{eq:4-1}
  H_q(A\tau)
  \ &= \ H_q(A\tau | v_{J_1},\dots,v_{J_{k+1}}, \alpha_1,\dots,\alpha_{k+1}) + 
  \sum_{i=1}^{k+1} I_q(A\tau; v_{J_i}, \alpha_i | v_{J_{i+1}},\dots, v_{J_{k+1}},\alpha_{i+1},\dots,\alpha_{k+1}) \nonumber\\
  \ &= \ \sum_{i=1}^{k+1} I_q(A\tau; v_{J_i}, \alpha_i | v_{J_{i+1}},\dots, v_{J_{k+1}},\alpha_{i+1},\dots,\alpha_{k+1}) ,
\end{align}
where the second equality uses that $J_1,\dots,J_{k+1}$ form a partition of $[L']$, so $A\tau$ is completely determined by $v_{J_1}, \ldots, v_{J_{k+1}}, \alpha_1, \ldots, \alpha_{k+1}$, and thus $H_q(A\tau|v_{J_1},\dots,v_{J_{k+1}},\alpha_1,\dots,\alpha_{k+1})=0$. For clarity, we note that the summand above when $i=k+1$ is simply $I_q(A\tau; v_{J_{k+1}}, \alpha_{k+1})$.
We thus have
\begin{align}
  H_q(A\tau)
  \ &\ge \ \sum_{i=1}^{k+1} I_q((A\tau)_{J_i}; v_{J_i}, \alpha_i | v_{J_{i+1}},\dots, v_{J_{k+1}},\alpha_{i+1},\dots,\alpha_{k+1}) \nonumber\\
  \ &= \ \sum_{i=1}^{k+1} I_q\left(v_{J_i} + \sum_{j \ge  i}^{} w\ind{j}_{J_i}\cdot \alpha_j; v_{J_i}, \alpha_i \bigg\vert v_{J_{i+1}},\dots, v_{J_{k+1}},\alpha_{i+1},\dots,\alpha_{k+1}\right) \nonumber\\
  \ &= \ \sum_{i=1}^{k+1} I_q\left(v_{J_{i}} + w\ind{i}_{J_i}\cdot \alpha_i; v_{J_i}, \alpha_i \big\vert v_{J_{i+1}},\dots, v_{J_{k+1}},\alpha_{i+1},\dots,\alpha_{k+1}\right) \nonumber\\
  \ &= \ \sum_{i=1}^{k+1} I_q\left(v_{J_{i}} + w\ind{i}_{J_i}\cdot \alpha_i; v_{J_i}, \alpha_i \right) \nonumber\\
  \ &= \ \sum_{i=1}^{k+1} H_q\left(v_{J_{i}} + w\ind{i}_{J_i}\cdot \alpha_i \right)
\end{align}
The inequality applies the data processing inequality to \eqref{eq:4-1}, using that $(A\tau)_{J_i}$ is a function of $A\tau$.
The first equality uses \eqref{eq:atau} and that $w\ind{1},\dots,w\ind{i-1}$ have no support in $J_i$ by definition of $J_i$.
The second equality uses that $\alpha_{i+1},\dots,\alpha_{k+1}$ are being conditioned on.
The third equality uses that the $v_i$'s and $\alpha_i$'s are all independent and that the $J_i$ are pairwise disjoint, so changing $v_{J_{i+1}},\dots,v_{J_{k+1}},\alpha_{i+1},\dots,\alpha_{k+1}$ does not affect $v_{J_{i}} + w\ind{i}_{J_i}\cdot \alpha_i$.
The last equality uses that $H_q(v_{J_i}+w\ind{i}_{J_i}\cdot \alpha_i | v_{J_i},\alpha_i) = 0$.
As $L\ge L_{p,q,\delta}$ and as $w\ind{i}_{J_i}$ has all nonzero entries by definition of $J_i$, we may apply Lemma~\ref{lem:ber-p} with $v = v_{J_i}$ and $\alpha=\alpha_i$ and $w = w\ind{i}_{J_i}$ and $d=|J_i|$.
This gives
\begin{align}
  H_q(A\tau)
  \ &\ge \ \sum_{i=1}^{k+1} H_q\left(v_{J_{i}} + w\ind{i}_{J_i}\cdot \alpha_i \right) \nonumber\\
  \ &\ge \ \sum_{i=1}^{k+1} |J_i|\cdot \left(h_q(p) + \frac{h_q(p)}{L+\delta}\right) \nonumber\\
  \ &= \ L'\cdot \left(h_q(p) + \frac{h_q(p)}{L+\delta}\right),
\end{align}
as desired.
The last equality uses that $J_1,\dots,J_{k+1}$ partition $[L']$. 
\end{proof}

\section{Upper bounds for average-radius list-decoding over $\F_2$}\label{sec:upper_bd}
In this section we prove the following theorem. Recall that we abbreviate $h(p)=h_2(p)$.

\begin{theorem}\label{thm:main-avg-rad}
	Let $n \in \N$. Let $p \in (0,\frac 12)$ and $R = 1-h(p)-\eps$, where $0<\eps<1-h(p)$. Let $L = \lfloor\frac{h(p)}\eps + 2\rfloor$. Then, a random linear code $\cC \leq \F_2^n$ of rate $R$ is $(p,L)$-average-radius list-decodable with probability $1-2^{-\Omega_{p,\eps}(n)}$.
\end{theorem}

Recall from the introduction that, following the techniques in \cite{GuruswamiHSZ02} and \cite{LiWootters}, we imagine sampling independent and uniform vectors $b_1,\dots,b_k$ and constructing the ``intermediate'' random linear codes $\cC_i = \spn\{b_1,\dots,b_i\}$. A potential function based argument is used to show that, with high probability, each of these intermediate codes is indeed $(p,L)$-average-radius list-decodable; in particular, this is true for $\cC_k$. 

Before discussing our potential function, we first briefly review the techniques of \cite{GuruswamiHSZ02} and \cite{LiWootters}; in particular, we describe the potential function they use. First, for a code $\cC$ and a vector $x \in \F_2^n$, we define 
\[
	L_\cC(x) := |\{c \in \cC: \delta(x,c) \leq p\}| \ .
\]
In \cite{GuruswamiHSZ02}, the authors define
\[
	S_{\cC} := \frac{1}{2^n}\sum_{x \in \F_2^n}2^{\eps n L_\cC(x)}
\]
and observe that, for any $b_1,\dots,b_i \in \F_2^n$, 
\[
	\Eover{b_{i+1}\sim \F_2^n}{S_{\cC_i+\{0,b_{i+1}\}}} = S_{\cC_i}^2 \ ,
\]
where we recall $\cC_i=\spn\{b_1,\dots,b_i\}$.\footnote{Here and throughout, for two subsets $A,B \subseteq \F_2^n$, we denote $A+B=\{a+b:a\in A, b \in B$. Thus, $\cC_i+\{0,b_{i+1}\} = \cC_{i+1}$.} That is, the potential function squares in expectation, so the probabilistic method guarantees that we can choose some $b_{i+1}$ for which $S_{\cC_{i+1}}\leq S_{\cC_i}^2$. 
Thus, for some choice of $b_1,\dots,b_k$, one has $S_{\cC_k} \leq (S_{\{0\}})^{2^k}$. 

In \cite{LiWootters}, the definition of $S_{\cC}$ is slightly modified: 
\[
	S_{\cC} := \frac{1}{2^n}\sum_{x \in \F_2^n}2^{\frac{\eps n L_{\cC}(x)}{1+\eps}} \ .
\]
This little bit of extra room allows to show that, in fact, with high probability over the choice of $b_{i+1}$, $S_{\cC_i + \{0,b_{i+1}\}} \leq S_{\cC_i}^2$. By a union bound, it follows that with high probability, $S_{\cC_k} \leq (S_{\{0\}})^{2^k}$. 

In either case, to conclude the proof, one observes the bound\footnote{Actually, for the potential function in \cite{LiWootters}, one has $S_{\{0\}} \leq 1 + 2^{-n(1-h(p)-\frac{\eps}{1+\eps})}$, but this difference does not matter for the conclusion.} $S_{\{0\}} \leq 1 + 2^{-n(1-h(p)-\eps)}$ and then uses
\[
	S_{\cC_k} \leq (S_{\{0\}})^{2^k} \leq (1 + 2^{-n(1-h(p)-\eps)})^{2^k} 
\leq \exp{ 2^{k-n(1-h(p)-\eps)} } \leq O(1)
\]
for $k$ chosen as above.

\subsection{Alterations for average-radius list-decoding}

While this argument analyzes the (absolute-radius) list-decodability of random linear codes very effectively, it is not immediately clear how to generalize the argument to study average-radius list-decodability. We now introduce the additional ideas we need to derive Theorem~\ref{thm:main-avg-rad}. We will fix a threshold parameter $\lambda \in (p,\frac{1}{2})$ for which $h(\lambda) < 1-R = h(p) + \eps$, to be determined later, and define 
\[ \eta \defeq 1-R-h(\lambda).\]  

We define the function $M_{R,\lambda}:[0,1] \to \R$ by 
\[
	M_{R,\lambda}(\gamma) := \begin{cases}
		1-R-h(\gamma) & \text{if } \gamma < \lambda \\
		0 & \text{if } \gamma \geq \lambda 
	\end{cases} \ .
\]
\begin{remark} \label{rmk:entropy-change}
	One can think of this quantity as a sort of ``normalized entropy change'' up to the threshold $\lambda$. 
Recalling that $1 - R = h(p) + \eps$, if $\gamma < \lambda$, then
\[ M_{R,\lambda}(\gamma) \approx \frac{1}{n}\inparen{ h(p) - h(\lambda)} \approx \log \inparen{ \frac{|B^n(p)|}{|B^n(\gamma)|}}, \]
where $B^n(p)$ denotes the Hamming ball in $\F_2^n$ of radius $p$.
Hence, $M_{R,\lambda}(\gamma)$ is something like a normalized ``surprise'' an observer would experience if they are expecting a random vector of weight $\leq p$ and see a vector of weight $\leq \gamma$. 
\end{remark}

For a linear code $\cC \leq \F_2^n$ and $x \in \F_2^n$ we define
\[
	L_{\cC,R,\lambda}(x) := \sum_{y \in \cC}M_{R,\lambda}(\delta(x,y)) .
\]
This is intuitively the ``smoothed-out'' list-size of $x$, where nearby codewords are weighted more heavily than far away codewords, and the weighting is given by the ``entropy change'' implied by the distance from $x$ to $y$. 

Next, we define 
\[
	A_{\cC,R,\lambda}(x) := 2^{\frac{nL_{\cC,R,\lambda}(x)}{1+\eta}}
\]
and
\[
	S_{\cC,R,\lambda} := \frac{1}{2^n}\sum_{x \in \F_2^n}A_{\cC,R,\lambda}(x) .
\]
The quantity $S_{\cC,R,\lambda}$ is the potential function we will analyze. 

\subsection{Proof of Theorem~\ref{thm:main-avg-rad}}

In this subsection we prove Theorem~\ref{thm:main-avg-rad}. The quantities $R$ and $\lambda$ (and hence $\eta = 1-R-h(\lambda)$) will be fixed throughout---although the precise value of $\lambda$ will be determined later---and so we will suppress their dependence and simply write $M(x)$, $L_{\cC}(x)$, $A_{\cC}(x)$ and $S_{\cC}$. 

First, we observe that the following analog of \cite[Lemma~3.2]{LiWootters} holds. The proof is a simple adaptation of theirs (which in turn follows \cite{GuruswamiHSZ02}).

\begin{lemma} \label{lem:list-doubles}
	For all $\cC \leq \F_2^n$ and $b \in \F_2^n$, 
	\begin{align}
		L_{\cC + \{0,b\}}(x) &\leq L_{\cC}(x) + L_{\cC}(x+b) \label{eq:L-bound}, \\
		A_{\cC+\{0,b\}}(x) &\leq A_{\cC}(x) \cdot A_{\cC}(x+b) \label{eq:A-bound} .
	\end{align}
	Moreover, equality holds if and only if $b \notin \cC$. 
\end{lemma}
\begin{proof}
    We have
    \begin{align}
        L_{\cC+\{0,b\}}(x)
        &\le   \sum_{y\in \cC} M(\delta(x,y))+M(\delta(x,y+b))\nonumber\\
        &=   \sum_{y\in \cC} M(\delta(x,y))+M(\delta(x+b,y))
        = L_{\cC}(x) + L_\cC(x+b),
    \end{align}
    and equality holds in the first line if and only if $\cC \cap (b+\cC) = \emptyset$, or, equivalently, $b \notin \cC$. The second inequality of the lemma statement follows from the first.
\end{proof}

Next, we bound $S_{\{0\}}$. We have 
\begin{align*}
	S_{\{0\}} &\leq 1 + 2^{-n}\sum_{\substack{x \in \F_2^n \\ \wt x \leq \lambda}} 2^{\frac{n \cdot (1 - R - h(\wt x))}{1+\eta}} \\
	&\leq 1 + \sum_{i=0}^{\lfloor \lambda n \rfloor} 2^{-n\inparen{1-h\inparen{i/n}-\frac{h(\lambda)+\eta-h\inparen{i/n}}{1+\eta}}} .
\end{align*}
As this sum is dominated by its last term, we deduce 
\begin{align} \label{eq:S_0-bound}
	S_{\{0\}} \leq 1 + (\lambda n)2^{-n\inparen{1-h(\lambda)-\tfrac{\eta}{1+\eta}}}.
\end{align}

From here, we can combine Lemma~\ref{lem:list-doubles} and \eqref{eq:S_0-bound} to deduce
\begin{restatable}{lemma}{SkBound} \label{lem:S_k-bound}
	Let $p \in (0,\frac 12)$ and $R = 1-h(p)-\eps$ for $0 < \eps < 1-h(p)$. Let $\cC_{Rn} \leq \F_2^n$ be a random linear code of rate $R$. Then $S_{\cC_{Rn}} \leq 2$ with probability at least $1-\exp{-\Omega_{\eta}(n)}$.
\end{restatable}

The proof of this lemma is completely analogous to that of \cite[Lemma~3.3]{LiWootters}. 
One only needs to be careful about the growth rate of $S_{\cC}$.
In particular, this proof crucially uses that $\eta$ is positive.
We again choose vectors $b_1,\dots,b_{Rn}$ independently and uniformly at random.
If $\cC_i = \spn\{b_1,\dots,b_i\}$, we need ``in expectation'' that $S_{\cC_i}\le 1+2^{-\Omega(n)}$ for all $i$ for the error bounds to succeed.
As we expect the $o(1)$ term to roughly double, we need $2^{Rn}\cdot (S_{(0)}-1) = 2^{-n(\eta - \frac{\eta}{1+\eta})} \le 2^{-\Omega_\eta(n)}$. For completeness, we provide the proof of Lemma~\ref{lem:S_k-bound} in Appendix~\ref{app:deferred-proof}.

Thus, in order to conclude Theorem~\ref{thm:main-avg-rad}, we are simply required to demonstrate that $S_{\cC} \leq 2$ implies that $\cC$ is $(p,L)$-average-radius list-decodable: this is the crux of our contribution. The main lemma we require is the following. 

\begin{lemma} \label{lem:S_C-to-entropy}
	Let $\cC\leq \F_2^n$ be a linear code of rate $R$ such that $S_{\cC} \leq 2$. Then, for all $x \in \F_2^n$ and $D \subseteq \cC \cap B(x,\lambda)$, it holds that
	\[
		\sum_{y \in D}h(\delta(x,y)) \geq (|D|-1-\eta)(1-R) - \frac{1+\eta}{n} .
	\]
\end{lemma}
\begin{proof}
	First, observe that for any $x \in \F_2^n$,
	\begin{align}
		L_{\cC}(x) &\geq \sum_{y \in D}\left((1-R)-h(\delta(x,y))\right) = |D|(1-R) - \sum_{y \in D}h(\delta(x,y)) \nonumber\\
		\text{so}\qquad
		\log A_{\cC}(x) &\geq n\frac{|D|(1-R) - \sum_{y \in D}h(\delta(x,y))}{1+\eta} .
		\label{eq:LD-bound}
	\end{align}
	Next, as $\delta(x,y) = \delta(x+z,y+z)$ for any $z \in \F_2^n$, we have, for any $x\in\mathbb{F}_2^n$ and $c\in \cC$, that $L_{\cC}(x)=L_{\cC}(x+c)$ and hence $A_{\cC}(x)=A_{\cC}(x+c)$.
	Thus, $\max_{x\in \mathbb{F}_2^n}A_\cC(x)$ is attained at at least $|\cC|$ different values of $x$, so
	\[
		S_{\cC} = \frac{1}{2^n}\sum_{x \in \F_2^n}A_{\cC}(x) \ge \frac{1}{2^n}\cdot |\cC|\cdot\max_{x\in \mathbb{F}_2^n}A_\cC(x)
		= 2^{-(1-R)n}\cdot \max_{x\in \mathbb{F}_2^n}A_\cC(x).
	\]
	Combining this with \eqref{eq:LD-bound}, we have, for any $x \in \F_2^n$, 
	\begin{align*}
		1 &\geq \log S_{\cC} \geq -(1-R)n + \log\inparen{A_{\cC}(x)} \\
		&\geq n \cdot \inparen{-(1-R) + \frac{|D|(1-R) - \sum_{y \in D}h(\delta(x,y))}{1+\eta}} \\
		&= n \cdot \frac{(|D|-1-\eta)(1-R) - \sum_{y \in D}h(\delta(x,y))}{1+\eta} .
	\end{align*}
	Rearranging yields the lemma.
\end{proof}

We may now conclude Theorem~\ref{thm:main-avg-rad}.

\begin{proof} [Proof of Theorem~\ref{thm:main-avg-rad}]
	Since $L = \lfloor\frac{h(p)}{\eps}+2\rfloor>\frac{h(p)}\eps+1 = \frac{1-R}\eps$, there exists $\eta>0$ small enough so that for all sufficiently large $n$
	\begin{align}\label{eq:L-condn}
		L > \frac{1-R+\eta+\frac{1+\eta}{n}}{\eps-\eta}.
	\end{align}
Thus, we define $\lambda$ so that $\eta$ (which we defined as $\eta = 1 - R - h(\lambda)$) 
satisfies \eqref{eq:L-condn}.
	Let $\cC$ be a random linear code of rate $R$. Due to Lemma~\ref{lem:S_k-bound}, the conclusion of Lemma~\ref{lem:S_C-to-entropy}, holds with probability $1-2^{-\Omega_{\eta}(n)}$ for $\cC$. It remains to show that, assuming $n$ is sufficiently large, any code $\cC$ satisfying the conclusion of Lemma~\ref{lem:S_C-to-entropy} is $(p,L)$-average-radius list-decodable. 
	
	Let $x \in \F_2^n$ and $\Lambda \subseteq \cC$ such that $|\Lambda| = L$; our goal is to show that, for all such $x$ and $\Lambda$, 
	\begin{align} \label{eq:avg-rad-goal}
		\frac{1}{L}\sum_{y \in \Lambda}\delta(x,y) > p .
	\end{align}
	Let 
	\[
		D = \{y\in \Lambda : \delta(x,y) \leq \lambda\}
	\]
	and define 
	\[
		h^*(\alpha) = \begin{cases}
			h(\alpha) & \text{if } \alpha \leq \frac 12 \\
			1 & \text{if } \alpha  > \frac 12
		\end{cases} .
	\]
	Now,
	\begin{align}
		\sum_{y\in \Lambda}h^*(\delta(x,y)) &\geq \sum_{y\in D}h(\delta(x,y)) + (L-|D|)h(\lambda) \label{eq:i}\\
		&\geq (|D|-1-\eta)(1-R)+(L-|D|)(1-R-\eta) - \frac{1+\eta}{n} \label{eq:ii}\\
		&=(1-R)\cdot (L-1) - \eta\cdot(1-R)-\eta\cdot(L-|D|) - \frac{1+\eta}{n} \nonumber \\
		&\geq (1-R)\cdot (L-1) -\eta\cdot (L+1) - \frac{1+\eta}{n} \nonumber \\
		&= (1-R)L - (1-R) - \eta\cdot(L+1) - \frac{1+\eta}{n} \nonumber \\
		&=Lh(p)-(1-R)-(L+1)\eta + L\eps - \frac{1+\eta}{n} \label{eq:iii}\\
		&=Lh(p)-(1-R+\eta)+L(\eps-\eta) - \frac{1+\eta}{n} \nonumber\\
		&>Lh(p) .\label{eq:iv}
	\end{align}
	Here, the Inequality~\eqref{eq:i} holds because $h^*(\alpha)> h(\lambda)$ for all $\alpha>\lambda$; Inequality~\eqref{eq:ii} is the conclusion of Lemma~\ref{lem:S_C-to-entropy}; Equality~\eqref{eq:iii} follows from the fact that $R = 1 - h(p) - \eps$; and Inequality~\eqref{eq:iv} follows from \eqref{eq:L-condn}. Thus, we deduce
	\begin{align} \label{eq:entropy-bound}
		\frac{1}{L}\sum_{y\in \Lambda}h^*(\delta(x,y)) > h(p) .
	\end{align}
	Since $h^*$ is concave,
	\[
		h^*\inparen{\frac{1}{L}\sum_{y\in \Lambda}(\delta(x,y))} \geq \frac{1}{L}\sum_{y\in \Lambda}h^*(\delta(x,y)) ,
	\]
	and so \eqref{eq:avg-rad-goal} follows from \eqref{eq:entropy-bound}, the monotonicity of $h^*$ and the fact that $h^*(p) = h(p)$. 
\end{proof}

\begin{remark} \label{rmk:rank-metric}
    Just as the argument in \cite{LiWootters} generalizes easily to the case of rank-metric codes, the same holds for the argument given above. Briefly, a rank-metric code is a set of matrices $\cC \subseteq \F^{m \times n}$, and the rank-distance between two matrices $X$ and $Y$ is $\delta_R(X,Y) = \frac{1}{n}\mathrm{rank}(X-Y)$ (where we assume without loss of generality that $n \leq m$). Using this notion of distance, one can again obtain a notion of list-decodability, and moreover \emph{average-radius} list-decodability. There is a ``rank-metric'' list-decoding capacity $R^*(p)$. \cite{LiWootters} showed that random linear rank-metric codes over the binary field $\F_2$ of rate $R^*(p)-\eps$ are with high probability $(p,(1-R^*(p))/\eps+2)$-list-decodable, and one can adapt the argument above to show that such codes are with high probability $(p,(1-R^*(p))/\eps+2)$-\emph{average-radius} list-decodable.
\end{remark}

\bibliographystyle{alpha}
\bibliography{refs}

\appendix

\section{Proof of Lemma~\ref{lem:S_k-bound}}\label{app:deferred-proof}
First, we restate Lemma~\ref{lem:S_k-bound} for the reader's convenience. 

\SkBound*

To prove Lemma~\ref{lem:S_k-bound}, we introduce the notation $T_\cC = S_\cC-1$ and show that if $T_\cC$ bounded away from 1, it doubles with sufficiently large probability whenever we add a uniformly random vector to $\cC$.

\begin{lemma}  \label{lem:T_C-markov}
    If $\cC \leq \F_2^n$ is a fixed linear code,
    \[
        \PRover{b\sim \F_2^n}{S_{\cC+\{0+b\}} > 1+2T_{\cC}+T_{\cC}^{1.5}} < T_{\cC}^{0.5} .
    \]
\end{lemma}

\begin{proof}
    Applying Lemma~\ref{lem:list-doubles}, for any fixed $b \in \F_2^n$, 
    \begin{align}
        S_{\cC} &= \Eover{x \sim \F_2^n}{A_{\cC+\{0,b\}}(x)}\\
        &\leq \Eover{x \sim \F_2^n}{A_\cC(x)A_{\cC+b}(x)} \\
        &= \Eover{x \sim \F_2^n}{-1+A_{\cC}(x)+A_{\cC}(x+b)+(A_{\cC}(x)-1)(A_{\cC}(x+b)-1)} \\
        &= 1 + 2T_{\cC} + \Eover{x \sim \F_2^n}{(A_{\cC}(x)-1)(A_{\cC}(x+b)-1)} . 
    \end{align}
    Now, if $x$ and $b$ are independent and uniformly random over $\F_2^n$, then so are $x$ and $x+b$, so we conclude
    \[
        \mathop{\mathbb{E}}_{b}\Eover{x}{(A_{\cC}(x)-1)(A_{\cC}(x+b)-1)} = \Eover{x}{(A_{\cC}(x)-1)}\Eover{b}{(A_{\cC}(x+b)-1)} = T_{\cC}^2.
    \]
    Hence, applying Markov's inequality (which is justified as $A_\cC(x)-1 > 0$ for all $x$),
    \begin{align*}
        \PRover{b\sim \F_2^n}{S_{\cC+\{0+b\}} > 1+2T_{\cC}+T_{\cC}^{1.5}} &\leq \PRover{b\sim \F_2^n}{\Eover{x}{(A_{\cC}(x)-1)}\Eover{b}{(A_{\cC}(x+b)-1)} > T_{\cC}^{1.5}} \\
        &< \frac{T_{\cC}^{1.5}}{T_{\cC}^2} = T_{\cC}^{0.5}. \qedhere
    \end{align*}
\end{proof}

We can now iteratively apply Lemma~\ref{lem:T_C-markov} to conclude Lemma~\ref{lem:S_k-bound}. 
\begin{proof} [Proof of Lemma~\ref{lem:S_k-bound}]
    Throughout the argument, we may assume $n$ is sufficiently large compared to $\eta$. First, for $i=0,1,\dots,k$, consider
    \begin{align*}
        \delta_0 := \lambda n 2^{-n(1-h(\lambda)-\frac{\eta}{1+\eta})},\\
        \delta_i := 2\delta_{i-1}+\delta_{i-1}^{1.5}.
    \end{align*}
    By induction, we claim that for all $i \leq n(1-h(p)-\eps)$, we have $\delta_i < 2^{i+1}\delta_0 < 2^{-\frac{\eta^2n}{3}}$.
    First, we note that for $k=(1-h(p)-\eps)n$,
    \begin{align}
        2^{k+1}\delta_0 &= (2\lambda n)\cdot 2^{n(1-h(p)-\eps)}\cdot 2^{-n(1-h(\lambda)-\frac{\eta}{1-\eta})} \nonumber\\
        &\leq (2\lambda n)\cdot 2^{n[(1-h(p)-\eps)-(1-h(\lambda)-\eta)-\frac{\eta^2}{2}]} \label{eq:app1}\\
        &= (2\lambda n)\cdot 2^{-\frac{\eta^2n}{2}} \label{eq:app2}\\
        &< 2^{-\frac{\eta^2n}{3}} \label{eq:app3}.
    \end{align}
    In the above, Inequality~\eqref{eq:app1} follows from the inequality $\frac{\eta}{1-\eta} \geq \eta + \eta/2$, valid for $\eta \in (0,1)$. We used the equality $h(p)+\eps = h(\lambda) + \eta$ to obtain \eqref{eq:app2}. The last line, \eqref{eq:app3}, holds for sufficiently large $n$. Hence, $2^{i+1}\delta_0 < 2^{-\frac{\eta^2 n}{3}}$ for all $i \leq k$, so we may assume this in inductively proving $\delta_i < 2^{i+1}\delta_0$ for all $i=0,1,\dots,k$. 
    
    Now, we clearly have $\delta_0 < 2\delta_0$ (so the base case of the induction holds), while for $i\geq 1$ we bound
    \[
        \delta_i = 2\delta_{i-1}\inparen{1+\tfrac{\sqrt{\delta_{i-1}}}{2}} = 2^i\delta_0\prod_{j=0}^{i-1}\inparen{1+\tfrac{\sqrt{\delta_j}}{2}} \leq 2^i\delta_0\cdot \exp{\frac{1}{2}\cdot\sum_{j=0}^{i-1}\sqrt{\delta_j}} < 2^{i+1}\delta_0.
    \]
    In the first two equalities, we applied the definitions of the $\delta_i$'s. The first inequality applies the estimate $1+z\leq e^z$, while the second uses the induction hypothesis $\delta_j<2^{-\frac{\eta^2 n}{3}}$ for $j<i$ and by ensuring $n$ is sufficiently large. 
    
    Now, let $b_1,\dots,b_k \sim \F_2^n$ be i.i.d. uniform random vectors, and let $\cC_i = \spn\{b_1,\dots,b_i\}$ denote the ``intermediate'' random linear codes. Call $\cC_i$ \emph{good} if $T_{\cC_i} \leq \delta_i$; we wish to show that with high probability, $\cC_i$ is good for all $i \leq k$. For $i=0$, we apply \eqref{eq:S_0-bound} and obtain
    \[
        T_{\cC_0} = S_{\{0\}}-1 \leq  1+\lambda n2^{-n(1-h(\lambda)-\frac{\eta}{1+\eta)}} - 1 = \delta_0.
    \]
    Now, let $i \geq 1$ and assume $\cC_i$ is good. By Lemma~\ref{lem:T_C-markov}, 
    \begin{align*}
        \PR{\cC_{i+1} \text{ is not good}} &= \PR{T_{\cC_{i+1}} > \delta_{i+1}} \leq \PR{T_{\cC_{i+1}} > 2T_{\cC_i} + T_{\cC_i}^{1.5}} \\
        &< T_{\cC_i}^{0.5} \leq \delta_i^{0.5}.
    \end{align*}
    Thus, with probability at least 
    \[
        1-\inparen{\delta_0^{0.5} + \delta_1^{0.5} + \cdots + \delta_k^{0.5}} > 1-k2^{-\frac{\eta^2n}{6}} \geq 1 - 2^{-\Omega_\eta(n)}
    \]
    we have $T_{\cC_i} \leq \delta_i$ for all $i=0,1,\dots,k$, as desired. In particular, we conclude $S_{\cC_k} = 1+T_{\cC_k} \leq 1+2^{-\frac{\eta^2 n}{3}} \leq 2$ with probability $1-2^{-\Omega_\eta(n)}$.
\end{proof}

\end{document}